\documentclass[10pt,journal]{IEEEtran} 
\IEEEoverridecommandlockouts
\usepackage{graphicx}
\usepackage{booktabs}
\usepackage{xcolor}

\usepackage{amsmath}
\usepackage{amssymb}
\usepackage{cite}
\usepackage{cuted}
\newcommand{\linebreakand}{%
  \end{@IEEEauthorhalign}
  \hfill\mbox{}\par
  \mbox{}\hfill\begin{@IEEEauthorhalign}
}    
\usepackage{bbm}
\usepackage{algorithm} 
\usepackage{algpseudocode}
\usepackage{caption}
\usepackage{steinmetz}
\usepackage{subcaption}
\usepackage{amsthm}
\usepackage[english]{babel}
\newtheorem{theorem}{Theorem}


\newcounter{relctr} %% <- counter for relations
\everydisplay\expandafter{\the\everydisplay\setcounter{relctr}{0}} %% <- reset every eq
\usepackage[figurename=Fig.]{caption}

\usepackage{xpatch}

\newtheoremstyle{remarkstyle}%
  {}%                       % Space above
  {}%                       % Space below
  {\itshape}%               % Body font
  {}%                       % Indent amount
  {\itshape}%              % Theorem head font
  {.}%                      % Punctuation after theorem head
  {.5em}%                   % Space after theorem head
  {}%                       % Theorem head spec (can be left empty, meaning `normal`)

\theoremstyle{remarkstyle}

\AtBeginDocument{} %% <- store original definition

\ifCLASSINFOpdf
\else
\fi
\begin{document}

\title{Sequential Monte Carlo for Resilient Networks: Assessment, Mitigation, and Generative Modeling
\thanks{The authors are with the University of Oulu, Finland. Emails: \{onel.alcarazlopez, amirhossein.azarbahram\}@oulu.fi. This work is supported in Finland by the Research Council of Finland (Grants 362782 and 369116 (6G Flagship)).}
}

\author{
Onel~L.~A.~L\'{o}pez, \emph{IEEE Senior Member} and Amirhossein~Azarbahram, \emph{IEEE Member}}

%\small	$^{*}$Centre for Wireless Communications (CWC), University of Oulu, Finland\\
%$^{\circ}$Department of Electrical and Electronics Engineering, Federal University of Santa Catarina, Florianopolis, Brazil \\
%$^{\dagger}$Department of Electronic Systems, Aalborg University, Denmark \\
%\small Emails: \{onel.alcarazlopez\}@oulu.fi, ??}

% The paper headers
% \markboth{Journal of \LaTeX\ Class Files,~Vol.~14, No.~8, March~2023}%
% {Shell \MakeLowercase{\textit{et al.}}: Bare Demo of IEEEtran.cls for IEEE Journals}

% make the title areane
\maketitle

\begin{abstract}

Resilience is becoming crucial for future wireless networks, which must withstand, adapt to, and recover from rare but potentially cascading disruptions. This paper develops a sequential Monte Carlo (SMC) simulation framework for such systems, in which resilience failures are formulated as path-dependent rare events arising from staged degradation and delayed recovery, and are decomposed into semantically interpretable levels defined by a reaction coordinate. 
Building on this structure, we present a fixed-level splitting approach with budget-aware population control, enabling efficient estimation of rare non-recovery probabilities. 
We discuss the potential reuse of SMC checkpoints as representative near-critical states for policy evaluation and simulation-based selection. 
We further extend the methodology to learned stochastic simulation by using generative sequence models as restartable surrogates within data-driven digital twins. 
We showcase the framework in a delay-critical wireless network use case, where SMC substantially improves over standard Monte Carlo in rare-event regimes with both physical and learned simulators. 
\end{abstract}

% Note that keywords are not normally used for peerreview papers.
\begin{IEEEkeywords}
assessment and control, generative models, path-dependent failure, resilience, sequential Monte Carlo
\end{IEEEkeywords}

\IEEEpeerreviewmaketitle

\section{Introduction}\label{sec:intro}

\IEEEPARstart{R}{esilience} is becoming a first-class requirement for the next generation of wireless communication systems \cite{Rak.2020,Matthiesen.2025,Alves.2025}.
Indeed, such future systems must operate as critical infrastructures, tightly coupled with other societal systems, where failures may propagate across domains and time scales, and must incorporate capabilities such as graceful degradation, rapid reconfiguration, and learning-driven recovery \cite{Alves.2025}. These capabilities naturally define a sequence of system operating stages, from nominal operation through degradation and critical regimes to eventual failure and recovery, as illustrated in Fig.~\ref{fig:stages}a, which must be explicitly accounted for to enable rigorous analysis and control.

\begin{figure}
    \centering
    \includegraphics[width=\linewidth]{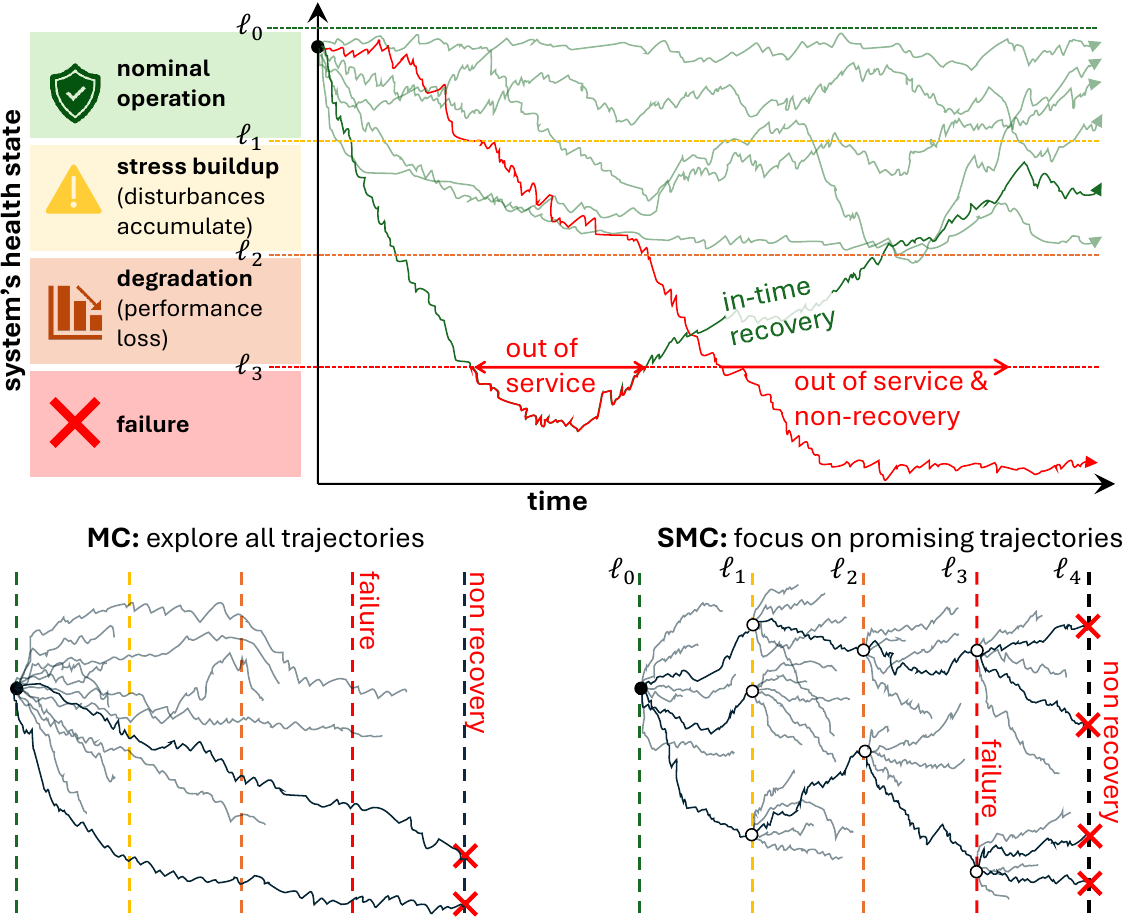}
    \caption{Illustration of system dynamics and trajectory sampling. (a) System evolution over time across health states, with some trajectories reaching the failure set, after which some recover and some do not (top). (b) Comparison of MC and SMC sampling strategies (bottom). Note that MC explores many full trajectories, with few reaching failure, whereas SMC concentrates effort on trajectories that progressively cross intermediate levels, replicating promising paths and efficiently identifying rare failure events.}
    \label{fig:stages}
    \vspace{-6mm}
\end{figure}

Simulation is crucial for resilience analysis and control, as it enables the exploration of complex system dynamics and the anticipation of failure regimes before they occur in practice. 
%This can run on digital twins (DTs), which are inherently data-driven (i.e., periodically updated through measurements from the physical system). Indeed, the collected data can be leveraged to learn stochastic simulation models and sampling of system dynamics.
This can run on data-driven digital twins (DTs), using measurements from the physical system to update stochastic models that randomly generate plausible system states.
Generative artificial intelligence (genAI) has recently emerged as a promising tool for this \cite{GenDT_network,DT_surrogate}. From a simulation perspective, the role of genAI is to learn a sampler for the stochastic evolution of system states, generating new state sequences without access to the physical system or relying on an exact analytical model. Thus, the learned model acts as a stochastic simulator, producing multiple plausible realizations of the system dynamics. For this task, variational autoencoders generate samples through latent-variable representations \cite{kingma2013auto}, normalizing flows construct invertible maps from simple base distributions to complex target distributions \cite{normalzing_flow_main}, generative adversarial networks learn a sampler through an adversarial game between a generator and a discriminator \cite{yoon2019time}, and diffusion models generate samples through iterative denoising from noise \cite{ho2020denoising}. 
%By leveraging these principles, DTs can act as learned stochastic simulators capable of producing multiple plausible network evolutions from the same initial condition.

Blindly sampling either a ground-truth simulator or a genAI surrogate in a DT might be an overkill in resilience studies. The reason is that classical Monte Carlo (MC) is prohibitively inefficient in capturing rare but impactful events \cite{Rubino.2009,Lopez.2023}, particularly those arising from sequential degradation or cascading failures \cite{Alves.2025}.
Indeed, efficient estimation requires methodologies that explicitly exploit the sequential structure of the progression toward failure, enabling targeted exploration of critical system trajectories. An important one is sequential MC (SMC) \cite{Lopez.2026}. 

SMC constructs intermediate probability distributions and approximates them using interacting particle systems, where trajectories are replicated as they approach increasingly critical regions of the state space \cite{Moral.2006,Rubino.2009,Smith.2013}. There are different SMC methods that leverage, e.g., importance splitting \cite{Villen.1991,Garvels.2000,Rubino.2009,Garvels.2002}, subset simulation \cite{Lopez.2023,Bect.2017}, and adaptive multilevel splitting \cite{Cerou.2007}. They all build on the fact that rare failures are a consequence of progressive transitions across levels and demonstrate substantial efficiency gains over standard MC. These aspects certainly make SMC particularly well-suited for systems exhibiting staged degradation and resilience dynamics. Fig.~\ref{fig:stages}b illustrates the sampling mechanisms of both MC and SMC side by side.
\vspace{-2mm}
\subsection{State-of-the-Art and Motivation}
Despite SMC's maturity, its application to network resilience studies remains largely limited, with only our recent conference paper \cite{Lopez.2026} initially exploring this explicitly, while other related works focus more on reliability/robustness assessment, e.g., \cite{Budde.2020,Lopez.2023}.\footnote{Notably, most rare-event simulation frameworks for wireless networks use non-sequential variants instead of SMC, e.g., \cite{Yu.2020,Ke.2023,Yu.2025}. Such approaches, although suitable for ultra-reliable low-latency communication systems, which is indeed their focus, do not fit well with the time/stage dynamics that must be accounted for in network resilience studies.} 
However, resilience involves a broader and more dynamic perspective, encompassing not only failure occurrences but also the evolution of system states across multiple regimes, including degradation, adaptation, and recovery. This highlights a gap between classical rare-event simulation methodologies and the requirements of resilience-aware system design. 

Meanwhile, genAI has recently been used in 
\cite{gibson2023flow_rare,gao2023rare,DASGUPTA2024109729_rare,NEURIPS2024_cfce7278_rare,STAMATELOPOULOS2025117589_rare,pandey2024heavy_rare} to improve rare-event simulations rather than just for synthetic data generation or forecasting.
This is mainly done by biasing samples toward rare regions, improving the representation of extreme events. For instance, normalizing flows have been used to induce bias through importance weights \cite{gibson2023flow_rare,gao2023rare,DASGUPTA2024109729_rare}, and unsupervised normalizing flows to generate non-local MC proposals for rare-event sampling \cite{Asghar2024_rare}. Moreover, diffusion-based models have been studied for tail-aware time-series generation and for assessing their ability to reproduce tail statistics \cite{NEURIPS2024_cfce7278_rare,STAMATELOPOULOS2025117589_rare,pandey2024heavy_rare}. 
%These works show that generative models can support rare-event analysis either by reshaping the sampling process toward informative low-probability regions or by improving the representation of tail trajectories. 
However, these studies do not directly address the path-dependent sampling structure required for resilience assessment, in which not only the faulty regime is of interest but also how the system evolves through degradation, adaptation, persistence, and recovery over time. Therefore, genAI in DTs for resilience should not only bias samples toward the rare-event set or reproduce tail statistics but also generate stochastic trajectory continuations from intermediate critical states. This motivates coupling genAI with SMC, where the learned model supplies plausible future continuations from different checkpoints, while SMC provides the level-wise rare-event exploration mechanism.

\subsection{Contributions}
Here, we build upon our preliminary conference work in \cite{Lopez.2026}, which 
introduces SMC for resilience assessment and control in wireless networks, to provide a more formal, detailed, and comprehensive treatment, including SMC estimator design, computational budgeting, mitigation-policy evaluation, and the use of learned generative surrogates within data-driven DTs. The specific contributions are fourfold:
\begin{itemize}
    \item We formalize an SMC-based framework for resilience assessment and control. Specifically, we formulate resilience assessment as a path-dependent rare-event estimation problem and show how fixed-level SMC can exploit the staged nature of degradation, critical service affectation, recovery, and non-recovery. Beyond \cite{Lopez.2026}, we formalize a budget-aware population-control mechanism that preserves the interpretability of semantically defined resilience levels while allocating effort to difficult transitions and by providing bias and variance scaling laws for the resulting estimator.

    \item We extend the framework toward data-driven DTs through generative sequence modeling, which is not considered in~\cite{Lopez.2026}. In this setting, the generative model acts as a restartable stochastic surrogate that can produce plausible future continuations from intermediate degraded states. This is different from using genAI only for nominal forecasting, synthetic data generation, or tail-statistics reproduction, since SMC requires conditional branching from checkpoints along the degradation path. We further discuss modeling choices for such surrogates, including single conditional generators, regime-specific generators, and mixture-of-experts architectures, together with criticality-aware data representation to improve learning of degradation, persistence, recovery, and near-failure transitions.

    \item We validate the proposed framework through a delay-critical wireless-network use case, illustrating resilience failure probability estimation, SMC-assisted service reconfiguration, and diffusion-based surrogate simulation. The numerical results show the efficiency gains of SMC over naive MC in rare non-recovery regimes, demonstrate how checkpoint-based policy selection can support resilience-aware reconfiguration, and provide evidence that a learned diffusion DT can be coupled with SMC to estimate rare resilience failures when the physical simulator is replaced by a generative surrogate.
\end{itemize}
%
%\textcolor{red}{organization, notations?}

% \textit{\textbf{Organization:}} The remainder of this paper is organized as follows. Section~\ref{sec:background} presents the preliminaries, including the generic dynamic time-evolving system, the SMC mechanism, and the associated estimation accuracy and simulation cost. Then, Section~\ref{sec:Res} discusses the use of SMC for resilience assessment and control, including simulation design and mitigation policies. Section~\ref{sec:AI} presents generative surrogate modeling for stochastic restartable simulation in data-driven DTs, together with criticality-aware data representation. Finally, Section~\ref{sec:system} provides a delay-critical wireless-network case study and numerical results demonstrating the performance of SMC and diffusion-based generative simulation, while Section~\ref{sec:conclude} concludes the paper.

\textit{\textbf{Organization:}} The remainder of this paper is organized as follows. Section~\ref{sec:background} introduces SMC within dynamical systems, which is then applied to resilience assessment and control in Section~\ref{sec:Res}. Section~\ref{sec:AI} 
extends the discussions to data-driven simulators for DTs. Section~\ref{sec:system} provides a delay-critical wireless-network case study and numerical results. Finally, Section~\ref{sec:conclude} concludes the paper.

\textit{\textbf{Notation:}} Bold lower/upper-case letters denote vectors/matrices. $\mathcal{O}(\cdot)$ denotes big-O complexity. $\max\{\cdot\}$ and $\min\{\cdot\}$ are maximum and minimum operators, respectively. $\mathbb{E}[\cdot]$ and $\text{Var}[\cdot]$ are respectively expectation and variance operators, while $\mathbb{E}'[\cdot]$ and $\text{Var}'[\cdot]$ are corresponding relative measures.
$\Pr(\cdot)$ and $\Pr(\cdot|\cdot)$ are probability and conditional probability operators, while $p(\cdot|\cdot)$ denotes a probability conditional law/function.
$\mathcal{N}(\boldsymbol{0},\mathbf{R})$ denotes a Gaussian distribution with mean $\boldsymbol{0}$ and covariance  $\mathbf{R}$. $\lfloor\cdot\rfloor$ and $\lceil\cdot\rceil$ are respectively floor and ceiling operators. $|\mathcal{A}|$ is the cardinality of the set $\mathcal{A}$.

\section{Preliminaries}\label{sec:background}
Consider a stochastic dynamical system $X(t)\in\mathcal{X}$, where $X(t)$ denotes the system state at time $t$ and $\mathcal{X}$ is a high-dimensional state space. Let $\mathcal{F}\subset \mathcal{X}$ denote a \emph{fault (or failure) set}, representing undesirable system configurations such as persistent service degradation or loss of recoverability. Then, define the failure event
\begin{align}
    \xi\triangleq \{\exists t\le T: X(t)\in\mathcal{F}\},\label{xi}
\end{align}
where $T$ is a finite observation horizon. Event $\xi$ corresponds to a hitting event, i.e., the system trajectory enters the fault set at least once during $[0,T]$. Importantly, $\xi$ is a path-dependent event, as it depends on the evolution of $X(t)$ over time rather than on its terminal state alone.

A fundamental task in resilience analysis is to estimate $p\triangleq \Pr(\xi)$, which quantifies the likelihood that the system ever transitions into an unacceptable or non-recoverable regime within the operational horizon. For this, we assume the availability of a stochastic simulator capable of generating sample paths of the system state. Specifically, given an initial condition $X(0)$, the simulator produces a realization
\begin{align}
    X^{(i)}(t)=\text{Sim}(X(0),\omega^{(i)}),\ t\in\{0,T\}, \label{eq:sim}
\end{align}
where $\omega^{(i)}$ denotes a realization of all exogenous randomness driving the system dynamics. Each call to the simulator with an independent $\omega^{(i)}$ yields an independent trajectory $X^{(i)}(t)$. 
The n\"aive MC approach to estimate $p$ lies in generating $N$ full independent trajectories and counting the fraction that enter the failure set within the observation horizon. This estimator is inefficient in the rare-event regime, since its relative variance scales as $(1-p)/(pN)$ and therefore the required number of full-trajectory simulations grows as $\mathcal{O}(1/p)$ when $p \ll 1$ \cite{Lopez.2023}. Instead, 
SMC addresses rare-event estimation by sequentially decomposing $\xi$ into a sequence of less rare conditional events, thereby avoiding the need to observe in full the evolution towards $\xi$ \cite{Cerou.2007,Lopez.2023}. 

\subsection{SMC Mechanism}
Let us introduce a reaction coordinate $g: \mathcal{X}\rightarrow \mathbb{R}$, which quantifies the system's progression toward the rare event set $\mathcal{F}$. Now, define an increasing sequence of thresholds (levels) $\ell_0 < \cdots<\ell_K$, with associated nested sets
\begin{align}
    \mathcal{L}_k \triangleq \{x\in \mathcal{X}: g(x)\ge \ell_k \},
\end{align}
such that $\ell_0=-\infty$, hence $\mathcal{L}_0=\mathcal{X}$, and $\mathcal{L}_K\subseteq \mathcal{F}$. The rare event probability can then be factorized as
\begin{align}
    \Pr(\xi)= \prod_{k=0}^{K-1}p_k,\label{pre}
    %\Pr(\mathcal{L}_1)\prod_{k=1}^{K-1}\Pr(\mathcal{L}_{k+1}|\mathcal{L}_k),
\end{align}
where $p_k\triangleq \Pr(\xi_{k+1}|\xi_k)$, 
%=\Pr(\mathcal{L}_{k+1}|\mathcal{L}_k)
$\xi_k\triangleq \{\exists t\le T: X(t)\in \mathcal{L}_k\}$ (and note that $p_0=\Pr(\xi_1)$). The aim is to estimate each factor $p_k$, and then recover $\Pr(\xi)$ via multiplication.

%SMC requires simulators to be restartable with independent continuation \cite{Villen.1991,Garvels.2000,Moral.2006}, i.e., able to resume from an intermediate state using fresh randomness. This holds for most event-driven and time-stepped simulators.
%In addition to generating independent trajectories, 
SMC requires the simulator to satisfy two mild properties \cite{Villen.1991,Garvels.2000}: i) restartability, meaning that the simulator can be paused at an intermediate time $t^\star$ and resumed from the current simulator state $X(t^\star)$; and ii) independent continuation, meaning that, upon restart, future system evolution can be generated using fresh realizations of the underlying randomness, independently across restarts. These properties are satisfied by the vast majority of event-driven and time-stepped simulators used in communication network studies. 
Mathematically, this is captured by the following, possibly implicit, update rule
\begin{align}
    X(t^\star+\Delta)=\text{Sim}(X(t^\star),W(t^\star)),\label{eq:sim2}
\end{align}
where $\Delta$ is a simulation time step and $W(t^\star)$ is a random input capturing all exogenous, independent across time and simulator runs, randomness over the interval $[t,t+\Delta]$, such that it can be generated independently of the past.

There are several SMC variants in the literature, differing in how levels are selected and how trajectories are replicated \cite{Villen.1991,Garvels.2000,Moral.2006,Rubino.2009,Smith.2013,Cerou.2007}. 
In the simplest form, referred to as fixed-level splitting, the thresholds $\{\ell_k\}$ are specified \emph{a priori}, and trajectories are split whenever a new level is reached, while adaptive multilevel splitting (AMS) dynamically determines levels based on the order statistics of sampled trajectories \cite{Cerou.2007,Brehier.2015}, typically retaining a fixed fraction of the most advanced trajectories at each iteration. Meanwhile, subset simulation combines Markov Chain MC with conditional sampling to estimate each stage's probability \cite{Au.2001,Lopez.2023}, and weighted splitting and sequential MC samplers let particles carry importance weights and perform resampling based on them \cite{Smith.2013}. 
Herein, we focus only on a fixed-level splitting approach, not only for simplicity but also because it is appealing for resilience studies, as discussed later in Section~\ref{sec:Res}. 
%More adaptive splitting strategies are left for future work.

%\subsection{Fixed-effort Multilevel Splitting}
%\subsubsection{SMC pseudocode}
Let $M$ denote the number of particles (trajectories) maintained at each level and $\xi_k^{(i)}\triangleq \{\exists t\le T: X^{(i)}(t)\in \mathcal{L}_k\}$. The algorithm is as follows:
\begin{itemize}
    \item \textbf{initialization:} generate $M$ 
    %i.i.d. trajectories $X^{(i)}(t)$ from the nominal initial distribution;
    i.i.d. initial simulator states $X^{(i)}(0)$ from the nominal initial-state distribution (or a deterministic initial condition);
    \item \textbf{level-by-level evolution:} for level $k=0,\cdots,K-1$, simulate each particle $i$ until it either reaches $\mathcal{L}_{k+1}$ (i.e., $\xi_k^{(i)}$ is non-empty) or $T$ (i.e., $\xi_k^{(i)}$ is empty). Upon hitting $\mathcal{L}_{k+1}$, store a checkpoint including all simulator state required for restarting in the subsequent splitting step;
    \item \textbf{selection and splitting:} let $S_{k}$ be the number of trajectories that reach $\mathcal{L}_{k+1}$ at level $k$. 
    Retain them and resample with replacement from their checkpoints to restore the population to $M$. 
    %Each clone is continued with independent randomness conditional on its checkpoint.
    Each resampled checkpoint initializes a new simulator instance, which is continued independently using fresh random input realizations.
    %These trajectories are retained, and $M-M_{k+1}$ new trajectories are generated by cloning the retained ones (with replacement), restarting from their stored hitting states, and evolving them with independent randomness.
    \item \textbf{estimation:} $\hat{p}_k=S_{k}/M$, and use \eqref{pre} such that
    \begin{align}
        \hat{p}_{\rm smc}=\prod_{k=0}^{K-1}\frac{S_{k}}{M},\label{smc}
    \end{align}
    which is unbiased under ideal i.i.d. sampling.
\end{itemize}

%SMC estimator \eqref{smc} is unbiased under i.i.d. sampling, while its 
\subsection{Accuracy and cost} The estimation variance depends primarily on the variance of the stage-wise conditional probabilities $p_k$. In fact, for moderate $\{p_k\}$ and  sufficiently large $M$, a standard, engineering approximation (ignoring inter-stage dependence due to resampling) for the relative variance yields \cite{Bect.2017}
\begin{align}
    \text{Var}'[\hat{p}_{\rm smc}] \approx \sum_{k=0}^{K-1}\frac{1-p_k}{p_k M}.\label{eq:var}
\end{align}
Meanwhile, the computational cost per level, in terms of the number of base-simulator time steps executed, is $C_k= \sum_{i=1}^{M}Q_{k,i}$, where $Q_{i,k}$ is the number of steps the particle $i$ is propagated while attempting to go from level $k$ to $k+1$. The total cost is then given by
\begin{align}
C_{\rm smc}=\sum_{k=0}^{K-1}C_k, \label{eq:Csmc}
\end{align}
Small values of $p_k$ lead to high variance, whereas large values call for more levels (and cost) and reduce interpretability.

\section{SMC for Resilience Assessment \& Control}\label{sec:Res}
Resilience in communication networks is fundamentally a trajectory-level property, as it reflects, or should reflect, not only whether service degradation occurs but also how the system evolves afterward, including its ability to recover within acceptable time and performance bounds \cite{Alves.2025}. As a result, resilience violations can be modeled as path-dependent rare events, rather than terminal-state failures. Specifically, a ``resilience'' failure event can be defined as in \eqref{xi}, where $\mathcal{F}$ denotes a set of unacceptable system states and may be time-dependent, e.g., encoding persistence or recovery-deadline constraints.
Such events are typically rare under nominal conditions but have a disproportionate impact when they occur, making their accurate estimation and control both challenging and essential. All this motivates the use of splitting-based techniques that explicitly exploit the progressive nature of resilience degradation. 
%Herein, we discuss how the basic framework from the previous section can be exploited for this.
%A ``resilience'' failure event can be defined using \eqref{xi} by making $\mathcal{F}$ to encode persistence and/or recovery-deadline constraints.
%capture delayed recovery, cascading degradation, and/or deadline-constrained service restoration.
%\footnote{The framework readily extends to time-dependent fault sets $\mathcal{F}(t)$.}
%Since such events are rare and path-dependent, splitting-based techniques that exploit progressive degradation are particularly suitable.
%This, and the fact that such events are typically rare under nominal conditions, motivate the use of splitting-based techniques that explicitly exploit the progressive nature of resilience degradation. 
%Herein, we discuss how the basic framework from Section~\ref{sec:smc} can be exploited for this.
%
\subsection{Simulation Design}
Key elements in SMC design/implementation are the definition of $g(\cdot)$, selection of hitting levels \cite{Garvels.2002,Cerou.2007}, and population control, as we discuss in the following. 
\subsubsection{Reaction coordinate} $g(\cdot)$ quantifies the system's progression toward the rare event set. 
In resilience assessment, this choice is particularly critical, as it determines not only estimator efficiency but also the interpretability of intermediate system states. 
Such a function must satisfy three key design principles:
\begin{itemize}
    \item expected monotonicity toward failure, such that 
    %trajectories that are more degraded or closer to non-recovery should, in expectation, exhibit larger function values;
    states closer to non-recovery should, on average, yield larger values of the coordinate \cite{Garvels.2002,Cerou.2007};\footnote{The reaction coordinate is not required to be strictly monotone along every realization, as its role is to induce a partial ordering of system states that enables a meaningful decomposition of the rare event into a sequence of less rare intermediate transitions.}
    \item early-warning capability, such that the coordinate should detect incipient degradation before irreversible failure occurs; and 
    \item operational interpretability, such that the coordinate should correspond to meaningful service-level or control-relevant quantities.
\end{itemize}

Note that $g(\cdot)$ is evaluated on the instantaneous simulator state $X(t)$, possibly including augmented state variables that encode temporal context. 
Natural reaction coordinate candidates in communication networks include backlog/delay levels,  accumulated stress or impairment measures, and recovery slack variables that quantify the remaining time before a recovery deadline is violated. Hybrid coordinates that combine state and temporal information are particularly useful when resilience violations are defined through persistence conditions rather than instantaneous thresholds.

\subsubsection{Hitting levels}
Aligning fixed levels with resilience phases (e.g., nominal operation, degradation, service level agreement violation, and non-recovery), as illustrated in Fig.~\ref{fig:stages},  promotes interpretability. Then, the corresponding $p_k$ quantifies stage-wise vulnerability, i.e., the likelihood that the system progresses from
one resilience phase to the next.
This framework enables an explicit identification of near-critical regimes, facilitating the systematic evaluation of reconfiguration and recovery policies from comparable system states, as discussed in Section~\ref{sec:reconf}. Indeed, it is particularly valuable when the goal is not only to estimate rare-event probabilities but also to understand how failures unfold and implement conscious mitigation and control strategies.

In general, the choice of level spacing/sets involves a trade-off between statistical efficiency and interpretability. This is because finely spaced levels reduce estimator variance but increase algorithmic complexity and may obscure the physical meaning of intermediate transitions, while overly coarse levels improve clarity but can lead to high variance or particle extinction if conditional transition probabilities are too small \cite{Garvels.2002,Cerou.2007}. In resilience-oriented studies, moderate conditional probabilities combined with meaningful phase boundaries may be preferable to aggressive variance minimization.

\subsubsection{Population control}\label{sec:PC}
Because the levels $\{\ell_k\}$ are defined a priori based on semantic notions of service degradation and persistence, the conditional probabilities $\{p_k\}$ may differ significantly. 
This can lead to particle extinction at difficult levels or inefficient computational effort allocation.
To address this, we adopt a budget-adaptive sampling strategy in which the total available computational cost $C_T$ is fixed a priori, and simulation effort is dynamically allocated across levels. 

\begin{algorithm}[t!]
\caption{Fixed-level SMC with Population Control}
\label{alg:budget_smc}
\begin{algorithmic}[1]
\Require Levels $\{\ell_k\}$, total budget $C_T$,
success target $S_{\rm tar}$, attempt target $A_{\rm tar}$,
initial pool size $M_0$, bounds $M_{m},M_{M}$,
safety factor $\varsigma$, probability floor $p_{\min}$

\State Initialize state pool $\mathcal P_0=\{x_0^{(1)},\ldots,x_0^{(M_0)}\}$
\State Set $C_{\rm used}\leftarrow 0$
\For{$k=0,\ldots,K-1$}
    \State Set $A_k\leftarrow 0$, $S_k\leftarrow 0$, $\mathcal S_k\leftarrow \emptyset$
    \State Set target level $L_{k+1}=\{x:g(x)\ge \ell_{k+1}\}$

    \While{$C_{\rm used}\!<\!C_T$ \textbf{and} $\left(S_k\!<\!S_{\rm tar}\ \textbf{or}\ A_k\!<\!A_{\rm tar}\right)$}
        \State Sample a checkpoint $x$ from $\mathcal P_k$ with replacement
        \State Propagate $x$ with fresh randomness until i) $X(t)\in L_{k+1}$, ii) the terminal horizon is reached, or iii) the absorbing failure event $\xi$ occurs
        \State Update $C_{\rm used}$ accoring to \eqref{eq:Cused}
        \State $A_k\leftarrow A_k+1$
        \If{the trajectory reached $L_{k+1}$}
            \State Store the first-hitting checkpoint in $\mathcal S_k$
            \State $S_k\leftarrow S_k+1$
        \EndIf
    \EndWhile

    \If{$A_k=0$ \textbf{or} $S_k=0$}
        \State \Return $\hat p_{\rm smc}=0$; \textbf{break}
    \EndIf

    \State Estimate $\hat p_k$ according to \eqref{eq:pkSA}

    \If{$k<K-1$}
        \State Choose next pool size according to \eqref{eq:Mk}
        \State Resample $M_{k+1}$ checkpoints with replacement from $\mathcal S_k$
        \State Set the resulting population as $\mathcal P_{k+1}$
    \EndIf

    \If{$C_{\rm used}=C_T$ \textbf{and} $k<K-1$}
        \State $\hat p_{\rm smc}=0$ (insufficient budget); \textbf{break}
    \EndIf
\EndFor

\State \Return $\hat p_{\rm smc}$ according to \eqref{eq:psmc}
\end{algorithmic}
\end{algorithm}

Let $S_{\rm tar}$ and $A_{\rm tar}$ denote the minimum number of successful crossings and attempted continuations required per level, respectively. Starting from an input pool of simulator states at level $\mathcal{L}_k$, trajectory continuations are generated sequentially with fresh randomness until they either reach $\mathcal{L}_{k+1}$, terminate at the horizon, hit the failure event, or the global budget is exhausted. 
 At level $k$, we record the number of attempts $A_k$, the number of successful first crossings $S_k$, and the corresponding cost \eqref{eq:Csmc} but now with $C_k=\sum_{i=1}^{A_k}Q_{k,i}$.
 Sampling at level $k$ stops when
%\begin{align}
    $S_k\ge S_{\rm tar}$
    and
    $A_k\ge A_{\rm tar}$,
%\end{align}
unless the global budget constraint
\begin{align}
    C_{\rm used}\triangleq \sum_{r=0}^{k} C_r \le C_T \label{eq:Cused}
\end{align}
reaches the limit earlier.
The empirical transition probability at level $k$ is then estimated as
\begin{align}
    \hat p_k=\frac{S_k}{A_k}.\label{eq:pkSA}
\end{align}
The estimator must return zero if $S_k=0$ or $A_k=0$, indicating that the next level was not reached under the available budget. Otherwise, the successful first-hitting checkpoints are resampled with replacement to construct the input pool for the next level. We choose the size adaptively as
\begin{align}
    M_{k+1}
   \! =\!
    \min\!\left\{\!
    M_{M},
    \max\!\left\{\!
    M_{m},
    \left\lceil
    \frac{\varsigma S_{\rm tar}}{\max(\hat p_k,p_{\min})}
    \right\rceil
    \right\}
    \right\},\label{eq:Mk}
\end{align}
where $\varsigma\ge 1$ is a safety factor, $p_{\min}>0$ prevents excessive population growth when $\hat p_k$ is very small, and $M_{m}$ and $M_{M}$ impose practical lower and upper bounds on the pool size. Thus, difficult transitions with small $\hat p_k$ induce larger subsequent pools, while easy transitions require fewer particles.
%Specifically, trajectory continuations are generated sequentially at each level $\mathcal{L}_k$ from the current state pool until both the number of trajectories $S_k$ that successfully reach the next level $\mathcal{L}_{k+1}$ and the total number of attempts $A_k$ are collected. That is, until conditions $S_k\ge S_{\rm tar}$ and $A_k\ge A_{\rm tar}$ are simultaneously met or the global budget is exhausted earlier.
The final rare-event probability estimate retains the multiplicative splitting structure of \eqref{pre} as
\begin{align}
    \hat{p}_{\rm smc}
    =
    \prod_{k=0}^{K-1}\hat p_k
    =
    \prod_{k=0}^{K-1}\frac{S_k}{A_k}. \label{eq:psmc}
\end{align}
The overall procedure is presented in Alg.~\ref{alg:budget_smc}.

The procedure introduces some bias due to outcome-dependent stopping, as in other adaptive SMC schemes (see discussions in \cite{Moral.2006,Bect.2017}). The following result formalizes this and provides some accuracy scaling laws. 

\begin{theorem}
\label{thm:bias_var}
Assume that the stopping is dominated by the success threshold $S_{\rm tar}$ (i.e., $A_{\rm tar}$ is not active), then the relative bias and variance of the stage estimator scale as
\begin{align}
    \mathbb{E}'[\hat p_k]\triangleq \frac{\mathbb{E}[\hat p_k]-p_k}{p_k}
    \approx
    \frac{1-p_k}{S_{\rm tar}},
    \ \ 
    %\frac{\mathrm{Var}[\hat p_k]}{p_k^2}
    \mathrm{Var}'[\hat p_k]\approx
    \frac{1-p_k}{S_{\rm tar}}.\label{eq:EV}
\end{align}
\end{theorem}
\begin{proof}
Each trajectory continuation at level $k$ can be viewed as a Bernoulli trial with success probability $p_k$. Let $\{Y_i\}_{i\ge 1}$ be i.i.d.\ Bernoulli$(p_k)$ variables indicating whether the $i$-th continuation reaches $\mathcal{L}_{k+1}$, hence $S_k=\sum_{i=1}^{A_k} Y_i$.

If the stopping is dominated by the success condition $S_k \ge S_{\rm tar}$, then $A_k$ is approximately the time needed to observe $S_{\rm tar}$ successes. In this case, $A_k$ follows a negative binomial distribution with mean $\mathbb{E}[A_k]\approx S_{\rm tar}/p_k$ and variance $\mathrm{Var}[A_k]\approx S_{\rm tar}(1-p_k)/p_k^2$. The estimator can then be approximated as
%\begin{align}
    $\hat p_k \approx S_{\rm tar}/A_k$.
%\end{align}
Using a first-order expansion around $\mathbb{E}[A_k]$, one obtains
\begin{align}
    \mathbb{E}[\hat p_k]
    \approx
    p_k + \frac{p_k(1-p_k)}{S_{\rm tar}},\ \ \mathrm{Var}[\hat p_k]
    \approx
    \frac{p_k^2(1-p_k)}{S_{\rm tar}},\label{eq:EaV}
\end{align}
leading straightforwardly to \eqref{eq:EV}.
\end{proof}

When both stopping conditions are active, or when the global budget truncates the procedure, the same mechanism in the above proof applies, although the expressions become more involved. Overall, increasing $S_{\rm tar}$ and $A_{\rm tar}$ reduces both (the positive) bias and variance, making the estimator asymptotically consistent.
%
%In general, the bias can be controlled easily through relatively large values of $S_{\rm tar}$ and $A_{\rm tar}$, and 
Another bias-reducing approach is to check the stopping condition only after batches of continuations, rather than after each individual attempt. 

\begin{theorem}
\label{thm:adaptive_smc_product}
Under Theorem~\ref{thm:bias_var}'s assumptions and approximate independence across levels (as in \eqref{eq:var} \cite{Bect.2017}), the relative bias and variance of the adaptive SMC estimator scale as
\begin{align}
    \mathbb{E}'[\hat p_{\rm smc}]\!&=\frac{\mathbb{E}[\hat p_{\rm smc}]-p}{p}
    =
    \prod_{k=0}^{K-1}(1+\mathbb{E}'[\hat{p}_k])-1, \label{eq:Esmc}\\
    \mathrm{Var}'[\hat p_{\rm smc}]\!&
    =\!
    \prod_{k=0}^{K-1}\!
    \left((1\!+\!\mathbb{E}'[\hat{p}_k])^2\!+\!\mathrm{Var}'[\hat{p}_k]\right)
    \!-\!
    \prod_{k=0}^{K-1}(1\!+\!\mathbb{E}'[\hat{p}_k])^2.\label{eq:Vsmc}
\end{align}
\end{theorem}
\begin{proof}
Using $\mathbb{E}[\hat p_k]=p_k(1+\mathbb{E}'[\hat{p}_k])$ under the approximate independence assumption gives
\begin{align}
    \mathbb{E}[\hat p_{\rm smc}]
    =
    \prod_{k=0}^{K-1}p_k(1+\mathbb{E}'[\hat{p}_k])
    =
    p\prod_{k=0}^{K-1}(1+\mathbb{E}'[\hat{p}_k]),
\end{align}
which immediately leads to \eqref{eq:Esmc}. 
Meanwhile,
\begin{align}
    \mathbb{E}[\hat p_{\rm smc}^2]&=\mathbb{E}\Big[\prod_{k=0}^{K-1}\hat{p}_k^2\Big]=\prod_{k=0}^{K-1} \mathbb{E}[\hat{p}_k^2]\nonumber\\
    & =
    p^2
    \prod_{k=0}^{K-1}
    \left(\mathrm{Var}'[\hat{p}_k]+(1+\mathbb{E}'[\hat{p}_k])^2\right).\label{eq:yyy}
\end{align}
where the last step comes from using 
\begin{align}
    \mathbb{E}[\hat p_k^2]
    &=
    \mathrm{Var}[\hat p_k]+\mathbb{E}[\hat p_k]^2
    \nonumber\\
    &=
    p_k^2 \mathrm{Var}'[\hat{p}_k]+p_k^2(1+\mathbb{E}'[\hat{p}_k])^2,\label{eq:xxx}
\end{align}

Subtracting $\mathbb{E}[\hat p_{\rm smc}]^2
=
p^2\prod_{k=0}^{K-1}(1+\mathbb{E}'[\hat{p}_k])^2$ from $\mathbb{E}[\hat p_{\rm smc}^2]$ in \eqref{eq:yyy}, 
and dividing by $p^2$ yields the relative-variance expression in \eqref{eq:Vsmc}.
\end{proof}

For small $\mathbb{E}'[\hat{p}_k]$ and $\mathrm{Var}'[\hat{p}_k]$, \eqref{eq:Esmc} and \eqref{eq:Vsmc} yield the first-order approximation
\begin{align}
    \mathbb{E}'[\hat p_{\rm smc}]
    \!\approx\!
    \sum_{k=0}^{K-1} \mathbb{E}'[\hat{p}_k],
    \ \
    \mathrm{Var}'(\hat p_{\rm smc})
    \!\approx\!
    \sum_{k=0}^{K-1} \mathrm{Var}'[\hat{p}_k].
\end{align}
This implies that the overall estimation error accumulates approximately linearly across levels, making the SMC estimator's performance highly sensitive to poorly estimated stages. 
%Consequently, an effective design should aim at balancing the stage-wise transition probabilities and allocating computational effort so that all levels contribute comparably to the total error.

%
%Finally, note that 
Overall, instead of adapting the levels as in other SMC frameworks relying on resampling and dynamic allocation of computational effort, e.g., \cite{Moral.2006,Smith.2013,Bect.2017,Cerou.2007}, our proposed formulation explicitly constrains the total computational budget and distributes it across levels through data-driven stopping criteria. This provides a simple and practical mechanism to concentrate simulation effort on the most critical transitions without requiring prior tuning of population sizes, and more importantly, it retains the interpretability of fixed-level splitting, making it particularly suitable for resilience analysis.

\subsection{Mitigation Policies}\label{sec:reconf}
Herein, we discuss the potential of the SMC machinery for fault-mitigation policies, which are naturally incorporated in resilient systems. 
We consider policies as decision rules that modify the system's stochastic evolution based on the current state and time, affecting disturbance, recovery dynamics, and/or corrective actions. Indeed, resilience-oriented policies are state/phase-dependent, activating only after certain degradation levels are reached \cite{Alves.2025}. The key idea is to reuse the near-critical system states generated during splitting as representative starting points for controlled interventions.

Let $u\in\mathcal{U}$ denote a policy that may encompass preventive, reactive, or recovery actions. Under policy $u$, the system state evolves according to a controlled stochastic process $X^u(t)$, 
inducing a policy-dependent resilience failure probability
\begin{align}
    p(u)\triangleq \Pr(\xi|u). %\label{eq:pu}
\end{align}
%Note that when the policies modify the system dynamics after certain levels are reached, the SMC estimator remains unbiased for the policy-conditional probability \eqref{eq:pu}
%
Note that policies are causal,
%assumed to be fixed over the duration of each simulated trajectory and to act 
i.e., decisions at time $t$ may depend only on information available up to that time, and may affect: i) the intensity or distribution of exogenous disturbances, ii) the system's recovery or stabilization dynamics, iii) the availability or effectiveness of corrective actions, or iv) combinations thereof. 
From a simulation standpoint, this corresponds to propagating trajectories according to policy-dependent dynamics between successive level-hitting events. 
Formally, this can be represented as
\begin{align}
X(t^\star+\Delta)
=
\mathrm{Sim}_{u(t^\star,X(t^\star))}\left(X(t^\star),W(t^\star)\right).    
\end{align}

Provided that the controlled process remains Markovian with respect to the simulator state used for checkpointing, standard SMC guarantees continue to hold for $p(u)$.
Indeed, the level sets themselves must remain policy-invariant, as they define the event decomposition. This ensures that i) the estimator remains unbiased for $p(u)$, ii) stage-wise probabilities are comparable across policies, and iii) differences in estimated resilience performance can be unambiguously attributed to policy effects. 

{Given the above, the conditional probabilities
\begin{align}
    p_k(u)\triangleq \Pr(\xi_{k+1}|\xi_k,u)
\end{align}
quantify how a given policy influences the likelihood of transitioning between successive resilience phases, hence providing fine-grained insight into where and how a policy is effective and allowing policy comparison. 
This quantity is estimated via MC by branching each stored checkpoint at level $\mathcal{L}_k$ into multiple continuations under policy $u$ to the next level $\mathcal{L}_{k+1}$. These checkpoints correspond to simulator states at the first hitting time of $\mathcal{L}_k$ and thus represent samples from the conditional (first-hitting) distribution induced by the baseline dynamics.}
This allows for state-contingent policy selection at predefined resilience levels as illustrated in Fig.~\ref{fig:lookahead} and discussed below. 

When a trajectory first reaches a level policy-hosted $\mathcal{L}_k$ at state $x_k^\star$, a temporary decision-support procedure can be invoked to select a mitigation policy based on simulated future outcomes.
Concretely, for each candidate policy $u_i\in\mathcal{U}$, an inner simulation is performed starting from the checkpoint state $x_k^\star$, using either SMC or MC over some or all remaining levels $\mathcal{L}_{k+1},\cdots,\mathcal{L}_K$ to estimate $\{p_{k}(u_i)\},\cdots,\{p_{K-1}(u_i)\}$ and other system-relevant performance parameters. Based on these estimates, a policy
\begin{align}
    u^\star(x_k^\star)=\arg\min_{u_i\in\mathcal{U}} \mathcal{J}(u_i|x_k^\star)
\end{align}
is selected according to a predefined decision criterion $\mathcal{J}(\cdot)$.
An intuitive, modeling for $\mathcal{J}$ is
\begin{align}
    \mathcal{J}(u_i|x_k^\star)=\sum_{l=k}^{K'}\ln \hat{p}_l(u_i) + c_k(u_i), \label{eq:LJ}
\end{align}
where $c_k(\cdot)$ captures the system cost (according to some metric) of implementing policy $u_i$. The impact of the conditional probabilities is linearized in \eqref{eq:LJ} to align with the behavior of common cost performance metrics. Meanwhile, $K'$ denotes a lookahead level such that $k\le K'\le K-1$. For instance, in a myopic setting, one has $K' = k$ such that \eqref{eq:LJ} leads to
%
%myopic modeling here could be $\mathcal{J}(u_i|x_k^\star)=\hat{p}_k(u_i)+c_k(u_i)$, wherein $c_k(\cdot)$ captures the system cost (according to some metric) of implementing policy $u_i$. Note that in such an example, we just consider
the policy that just reduces the chances of progressing one level in the fault chain with a reasonable cost. Obviously, approaches with greater $K'$ should lead to better mitigation, but considering always that 
%long-sighted policies increase computational cost.
evaluating deeper lookahead (i.e., simulating across more remaining levels or longer horizons) increases computational cost. 
%Obviously, approaches with greater $K'$ can be designed, but considering always that 
%long-sighted policies increase computational cost.
%evaluating deeper lookahead (i.e., simulating across more remaining levels or longer horizons) increases computational cost.

Importantly, the simulations used for policy selection serve only as decision support and are not reused as part of the realized system evolution. After selecting $u^\star(x_k^\star)$, the system is resumed from the same checkpoint state $x_k^\star$  using fresh, independent randomness under the selected policy. 
Under this separation, policy selection does not violate causality since simulated futures inform the policy choice, but the actual post-decision evolution remains an independent realization of the controlled dynamics. As a result, the procedure defines a well-posed state-contingent policy that can be evaluated by MC or SMC at the outer level. The computational overhead of the inner simulations can be controlled by limiting the number of candidate policies or by invoking the selection procedure only at a small number of critical levels, as indicated earlier.

The above procedure enables local, state-contingent policy selection at predefined resilience levels through simulation-based lookahead. Note that this does not aim to perform global policy synthesis or solve a full optimal control problem over the system state space. Instead, the framework supports bounded, simulation-driven optimization over a finite set of candidate mitigation policies, applied selectively at near-critical operating points where resilience is most at risk. This allows SMC to provide a principled mechanism for informing policy decisions, particularly in regimes dominated by rare, path-dependent events, complementing analytical optimization and learning-based approaches that are typically designed for averaged or data-rich operating conditions.

\begin{figure}[t!]
    \centering
    \includegraphics[width=1\linewidth]{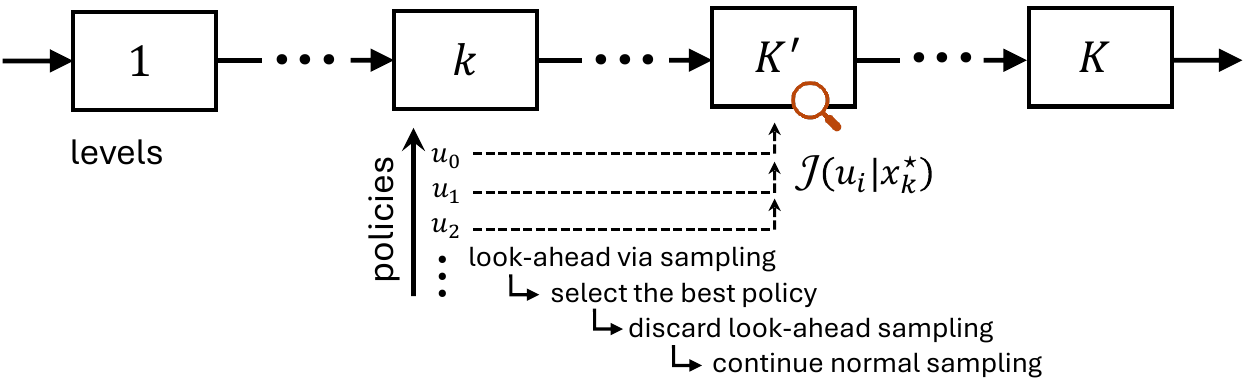}
    \caption{Illustration of the online mitigation/control procedure.}
    \label{fig:lookahead}
\end{figure}

%Although the above SMC procedure enables local, state-contingent policy selection at predefined resilience levels through simulation-based lookahead, it does not aim to perform global policy synthesis or solve a full optimal control problem over the system state space. Instead, the framework supports bounded, simulation-driven optimization over a finite set of candidate mitigation policies, applied selectively at near-critical operating points where resilience is most at risk. In this sense, 
%Overall, SMC provides a principled mechanism for informing policy decisions, particularly in regimes dominated by rare, path-dependent events, complementing analytical optimization and learning-based approaches that are typically designed for averaged or data-rich operating conditions.
%

\section{Stochastic Simulation with Learned Surrogates} \label{sec:AI}
%\textcolor{blue}{maybe title can be modified to give some hint on generative modeling}

So far, we have assumed access to an explicit and perfect stochastic simulator, which is rarely the case in real-world deployments. 
Indeed, operational systems are characterized by partial observability, high-dimensional and heterogeneous dynamics, and hidden couplings/interactions \cite{Alves.2025}. Therefore, any simulator, if available at all, can only provide an approximation of the true dynamics, no matter how accurate it can be. In practice, one typically has access to data rather than to the true physical model. The relevant task is therefore to learn a stochastic simulator from observations. A data-driven DT provides the operational form of this learned simulator, enabling trajectory sampling for resilience assessment without requiring access to the physical system or an exact model. Such a learned simulator must capture the temporal structure of the observed history and, autoregressively, generate uncertain future evolution until a critical event of interest, such as service degradation or failure, is reached.

% FOR INTRO:
% Data-driven sequence modeling provides a natural foundation for this task. Recurrent neural networks, long short-term memory (LSTM) networks \cite{hochreiter1997long}, attention-based recurrent architectures, and Transformers \cite{vaswani2017attention} have been widely used to capture temporal dependencies and inter-slot correlations in sequential data. In parallel, generative models such as variational autoencoders (VAEs) \cite{kingma2013auto} and diffusion models \cite{Ho.2020} have enabled stochastic sampling from complex, high-dimensional distributions. Combining temporal sequence models with probabilistic generative modeling enables DTs to generate diverse and physically plausible trajectory rollouts.

\subsection{Generative Sequence Modeling}

Imagine time is discretized and let $X[j]$ denote a state descriptor of the network at time index $j$. For a history length $L_h$ and a future rollout length $L_f$, define
\begin{subequations}\label{eq:futhist_concat}
\begin{align}
\mathbf{x}^{\mathrm{h}}_j
&=
\big[X[j-L_h+1],\ldots,X[j]\big]^T,
\label{eq:futhist_concat_h}\\
\mathbf{x}^{\mathrm{f}}_j
&=
\big[X[j+1],\ldots,X[j+L_f]\big]^T .
\label{eq:futhist_concat_f}
\end{align}
\end{subequations}
Here, $\mathbf{x}^{\mathrm{h}}_j$ summarizes the recent system evolution, while $\mathbf{x}^{\mathrm{f}}_j$ denotes the future block. In addition, let $\mathbf{z}_j$ denote side information relevant to how the future may evolve, such as the operating regime, control action, reaction-coordinate, persistence counter, or external context.

A generic sequence-to-sequence model first maps the history window and side information into a temporal representation given by
\begin{equation}\label{eq:encoderLR}
\mathbf{h}_j = \mathrm{Enc}_\phi(\mathbf{x}^{\mathrm{h}}_j,\mathbf{z}_j),
\end{equation}
where $\mathrm{Enc}_\phi(\cdot)$ may be implemented using recurrent units or attention-based models \cite{vaswani2017attention}. The future block is then produced from this representation by a generation module. To learn a restartable simulator, we learn the conditional continuation law
$
p_\theta\!\left(
\mathbf{x}^{\mathrm{f}}_j
\mid
\mathbf{x}^{\mathrm{h}}_j,\mathbf{z}_j
\right)$.
This law is induced by a stochastic generator of the form
\begin{equation}\label{eq:decoderLR}
\widehat{\mathbf{x}}^{\mathrm{f}}_j
=
\mathrm{Sim}_\theta(\mathbf{h}_j,\boldsymbol{\epsilon}),
\end{equation}
where $\boldsymbol{\epsilon}$ is a randomness source, ensuring different continuations from a given encoded state $\mathbf{h}_j$. Depending on the model class, $\boldsymbol{\epsilon}$ may correspond to a latent variable, injected diffusion noise, sampling noise, or other stochastic component. 

The way in which $\mathbf{z}_j$ is taken into account determines how the conditional trajectory law is parameterized. For this, the following modeling choices are relevant:

\begin{figure}
    \centering
    \includegraphics[width=\linewidth]{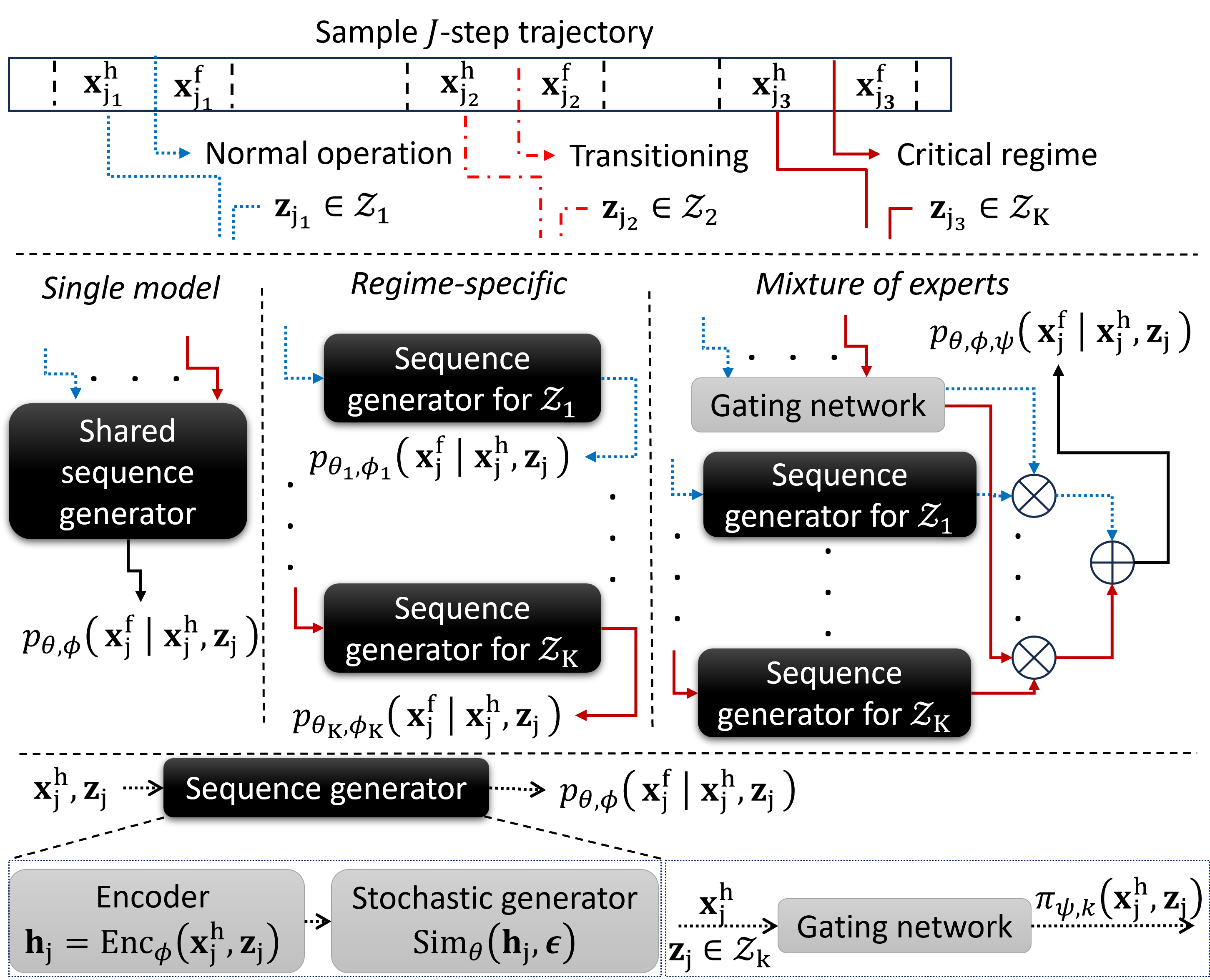}
    \caption{Sequence generation modeling choices for a restartable stochastic simulator. A trajectory is divided into history--future windows associated with different operating regimes. The conditional future distribution can be modeled using a single shared generator, regime-specific generators, or a mixture-of-experts architecture with a gating network.}
    \label{fig:learning_blocks}
\end{figure}

\begin{itemize}
    \item \textbf{Single conditional generator:}
    A single encoder--generator pair is used for all operating regimes. Here, the future block is sampled as in \eqref{eq:decoderLR}. Equivalently, this induces the conditional trajectory law
    $
    p_{\theta,\phi}
    \left(
    \mathbf{x}^{\mathrm{f}}_j
    \mid
    \mathbf{x}^{\mathrm{h}}_j,\mathbf{z}_j
    \right).
    $
    This option is suitable when different regimes share common temporal structure, e.g., when $\mathbf{z}_j$ is continuous, moderately varying, or when the corresponding regimes exhibit overlapping behavior. 

    \item \textbf{Regime-specific generators:}
    Distinct encoder--generator pairs can be trained for different values or ranges of $\mathbf{z}_j$. Let $\{\mathcal{Z}_k\}_{k=1}^{K}$ denote a partition of the context space, where each $\mathcal{Z}_k$ corresponds to a specific regime. If $\mathbf{z}_j\in\mathcal{Z}_k$, then the $k$-th model is used such that $
    \widehat{\mathbf{x}}^{\mathrm{f}}_j
    =
    \mathrm{Sim}_{\theta_k}
    (
    \mathbf{h}^{(k)}_j,\boldsymbol{\epsilon}
    ).
    $
    Hence, the corresponding conditional law is written as 
    \begin{equation}
    p_\theta
    \left(
    \mathbf{x}^{\mathrm{f}}_j
    \mid
    \mathbf{x}^{\mathrm{h}}_j,\mathbf{z}_j
    \right) = 
    p_{\theta_k,\phi_k}
    \left(
    \mathbf{x}^{\mathrm{f}}_j
    \mid
    \mathbf{x}^{\mathrm{h}}_j,\mathbf{z}_j
    \right),
    \ 
    \mathbf{z}_j\in\mathcal{Z}_k.
    \end{equation}
    This is relevant when the regimes are explicitly identifiable from $\mathbf{z}_j$, have limited overlap, and exhibit sufficiently distinct local transition laws to justify separate model parameters. In this case, a single model would average across modes. However, this approach partitions the dataset and requires sufficient data for each regime. It may also introduce discontinuities when switching between models during inference.

    \item \textbf{Mixture-of-experts generator:}
    A mixture-of-experts formulation uses several expert encoder--generator pairs together with a gating function \cite{VAE_MOE}. The $k$-th expert induces $p_{\theta_k,\phi_k}
    \left(
    \mathbf{x}^{\mathrm{f}}_j
    \mid
    \mathbf{x}^{\mathrm{h}}_j,\mathbf{z}_j
    \right)$, while the overall conditional trajectory law is then given by
    \begin{multline}
    p_{\theta,\phi,\psi}
    \left(
    \mathbf{x}^{\mathrm{f}}_j
    \mid
    \mathbf{x}^{\mathrm{h}}_j,\mathbf{z}_j
    \right)
    = \\
    \sum_{k=1}^{K}
    \pi_{\psi,k}
    \left(
    \mathbf{x}^{\mathrm{h}}_j,\mathbf{z}_j
    \right)
    p_{\theta_k,\phi_k}
    \left(
    \mathbf{x}^{\mathrm{f}}_j
    \mid
    \mathbf{x}^{\mathrm{h}}_j,\mathbf{z}_j
    \right).
    \end{multline}
    Here, $\pi_{\psi,k}(\mathbf{x}^{\mathrm{h}}_j,\mathbf{z}_j)$ is the gating weight of the $k$-th expert, satisfying
    \begin{equation}
    \pi_{\psi,k}
    \left(
    \mathbf{x}^{\mathrm{h}}_j,\mathbf{z}_j
    \right)\geq 0, \quad
    \sum_{k=1}^{K}
    \pi_{\psi,k}
    \left(
    \mathbf{x}^{\mathrm{h}}_j,\mathbf{z}_j
    \right)=1.
    \end{equation}
    Each expert can specialize in a different part of the operation phase, while the gate selects or combines experts according to the current history and context. This is useful when regime boundaries are soft or overlapping. Compared with fully separate regime-specific generators, it enables specialization while allowing smoother transitions between regimes. However, it introduces additional training complexity and requires careful regularization or balanced sampling to avoid expert collapse, especially when rare regimes are underrepresented.
\end{itemize}

The overall procedure for learning sequence generation models for stochastic simulation, together with the corresponding modeling choices, is illustrated in Fig.~\ref{fig:learning_blocks}. Once the conditional generator is learned, long-horizon trajectories can be generated autoregressively. Given the current history and context, a future block is sampled as
$
\widehat{\mathbf{x}}^{\mathrm{f}}_j
\sim
p_\theta
\left(
\mathbf{x}^{\mathrm{f}}_j
\mid
\mathbf{x}^{\mathrm{h}}_j,\mathbf{z}_j
\right).
$
The generated samples are appended to the history, the oldest samples are discarded, and the context variables are updated. This produces the rollout
\begin{equation}
\mathbf{x}^{\mathrm{h}}_j
\rightarrow
\widehat{\mathbf{x}}_{j+1:j+L_f}
\rightarrow
\widehat{\mathbf{x}}_{j+L_f+1:j+2L_f}
\rightarrow
\cdots,
\end{equation}
until the maximum horizon is reached or a stopping condition is satisfied. The stopping condition may correspond to recovery, a new critical level, or a resilience failure event.
%
%This window-based stochastic generator can therefore act as a restartable simulator. Specifically, 
Overall, this supports the two operations required by SMC: conditional continuation from an intermediate state and branching from a stored checkpoint.

\subsection{Criticality-Aware Data Representation}\label{sec:criticality_data}

% The model-side mechanisms discussed earlier in Section~\ref{sec:intro} mainly address how the generative model or sampling process can better represent rare regimes, but model expressiveness alone is not sufficient. If the training data do not expose the model to degradation, persistence, recovery, and failure-prone transitions, the learned model may reproduce average dynamics while ignoring the rare paths that dominate resilience failures. Thus, data representation must also be designed with tail behavior in mind. This becomes particularly important when rare transitions have low empirical frequency, as increasing the dataset size alone may still leave the model with too few informative near-critical samples for stable training.

The model-side mechanisms discussed in Section~\ref{sec:intro} mainly address how the generative model or sampling process can better represent rare regimes, but model expressiveness alone is not sufficient. If the training data do not expose the model to degradation, persistence, recovery, and failure-prone transitions, the learned model may reproduce average dynamics while ignoring the rare paths that dominate resilience failures. Thus, data representation must also be designed with tail behavior in mind. This improves both tail coverage and sample efficiency, such that rare transition windows become statistically visible during training, while abundant low-criticality samples do not dominate gradient updates.

Given a set of trajectories
\begin{equation}
    \mathcal{D}_{\mathrm{traj}}
=
\left\{
\mathbf{x}^{(n)}_{1:J_n}
\right\}_{n=1}^{N},
\end{equation}
a windowed training dataset can be constructed as
\begin{equation}
\begin{aligned}
\mathcal{D}
=
\Big\{
&\big(
\mathbf{x}^{(n)}_{j-L_h+1:j},
\mathbf{x}^{(n)}_{j+1:j+L_f},
\mathbf{z}^{(n)}_j
\big) \; \Big| \\
& n=1,\ldots,N,\quad
j=L_h,\ldots,J_n-L_f
\Big\}.
\end{aligned}
\end{equation}
Here, each sample contains a history window, a future window, and a conditioning vector. As discussed previously, $\mathbf{z}^{(n)}_j$ may include information relevant to resilience assessment, but including it in the conditioning signal is useful only if the corresponding regimes are sufficiently represented in the dataset. For this, uniform selection is usually insufficient because most windows are collected from low-criticality parts of the trajectories, leading to poor tail representation. To address this, windows can be grouped according to their criticality using the reaction coordinate $g(\cdot)$ and the levels $\{\ell_k\}$, such that
\begin{equation}
\mathcal{D}_k
=
\left\{
(\mathbf{x}^{\mathrm{h}}_j,\mathbf{x}^{\mathrm{f}}_j,\mathbf{z}_j):
\ell_k \leq g(x_j)<\ell_{k+1}
\right\}.
\end{equation}
The subsets $\{\mathcal{D}_k\}$ separate different regimes depending on the chosen level definitions in the system of interest.

Training batches can then be sampled from a level-aware empirical mixture given by
\begin{equation}
q(\mathcal{D})
=
\sum_{k=0}^{K}
\alpha_k q_k(\mathcal{D}_k),
\end{equation}
where $q_k$ is the empirical distribution over windows in $\mathcal{D}_k$, and $\alpha_k$ controls how often windows from level $k$ are selected during training. Choosing larger $\alpha_k$ for higher-criticality levels increases the exposure to rare but important transitions underrepresented by uniform sampling.

%Note that here, oversampling only rare failures may distort the learned dynamics and produce a poor surrogate for sequential rollout. 
The dataset should
%instead
preserve the conditional transition structure across levels, i.e., how trajectories enter degraded regimes, how long they persist near critical thresholds, and how they either recover or progress toward failure. This leads to tail-aware data representation, complementing tail-aware generative modeling. 
%We will provide an example of criticality-aware dataset construction later in the following Section.

\section{A Use Case Example}\label{sec:system}

We illustrate the proposed resilience framework from Sections~\ref{sec:Res} and \ref{sec:AI} through a delay-critical wireless network system, e.g., a user(s)-RAN-transport-core-application segment, subject to stochastic stress and recovery dynamics. 
The idea is not to model individual system or protocol layers but to capture service-level degradation and recovery behavior, which is paramount for resilience assessment.
\subsection{System Model}\label{sec:model}
The system is observed over a finite horizon $T$, discretized into time steps of duration $\Delta=T/J$ and indexed by $j=0,\cdots,J$. We consider a service queue fed by a constant normalized arrival workload $\Lambda$,
%\in(0,1)
expressed relative to nominal capacity, and assumed constant to isolate resilience effects induced by network stress and recovery rather than by traffic fluctuations. 
This corresponds with a steady service demand and simplifies the interpretation of rare-event probabilities. Moreover, the normalized instantaneous service capacity at time step $j$ is denoted by $C[j]\in(0,1)$,
%and evolves with the system's degradation state, 
while $B[j]$ is the corresponding backlog, representing the amount of unfinished work measured in equivalent service time and evolving as
%The system is characterized by a constant offered traffic rate $\Lambda$, a time-varying effective service rate $C[k]$, and an aggregate backlog $B[k] \ge 0$, where $k$ denotes a discrete-time index. 
%The backlog evolves as
%
\begin{align}
    B[j+1]=\max\{0, B[j]+(\Lambda-C[j])\Delta\}.\label{eq:B}
\end{align}
%where $\Delta$ is the time block duration.
%with initial condition $B[0]=0$.
This single backlog state aggregates congestion effects across the end-to-end path.

%The effective service rate is modeled as $C[k]=C_0\eta[k]$, where $C_0$ is the nominal service capacity and $\eta[k]\in[\eta_0,1]$ is a capacity degradation factor, and $\eta_0>0$ ensures residual service availability.

The service capacity dynamics are driven by environmental stress and accumulated degradation, with a logistic mapping capturing saturation effects that are common in nonlinear systems, e.g., \cite{Akhter.2023}. Specifically, we adopt
\begin{align}
    C[j]&\triangleq \frac{1}{1+e^{-\eta[j]}}, \label{eq:C} 
\end{align}
where $\eta[j]\in\mathbb{R}$ is a latent health state given by
\begin{align}
    \eta[j+1]&=\eta[j]+\nu(u)(1-C[j])^{\varphi(u)} - F'[j].\label{eq:eta}
\end{align}
Here, $\nu(u)>0$ is the recovery rate and $\varphi(u)>1$ controls the nonlinearity of the recovery dynamics, both controlled by a recovery acceleration action $u$.
The recovery rate models adaptive system response that increases when the system operates below capacity, while the exponential term captures nonlinear amplification of stress effects, consistent with models of cascading failures and overload phenomena in complex networks \cite{Rak.2020}. Meanwhile, $F'[j]\triangleq e^{F[j]}$ comprises log-normal fatigue effects using a standard autoregressive process to capture temporal correlation in environmental disturbances as
\begin{align}
    F[j+1] = \rho F[j] + (1-\rho)\mu_F + \gamma\sigma_F, \label{eq:relax}
\end{align}
where $\rho\in[0,1)$ controls temporal correlation, $\mu_F\in\mathbb{R}$ is the long-term mean of the latent stress level with $\sigma_F> 0$ determining its variability, and $\gamma\sim\mathcal{N}(0,1)$. 
%This produces a lognormal stress distribution. 

The end-to-end delay experienced by traffic is approximated using Little's law as 
\begin{align}
    %D[k]=D_0+\frac{B[k]}{C[k]+\epsilon},
    D[j]=B[j]/C[j]. \label{eq:delay}
\end{align} 
%which captures the sensitivity of delay to both backlog accumulation and service degradation. 
 %
This captures the sensitivity of delay to both backlog accumulation and service degradation. 
Refer to Fig.~\ref{fig:bufferS} for a high-level summary of the system dynamics. 
\begin{figure}[t!]
    \centering    \includegraphics[width=1\linewidth]{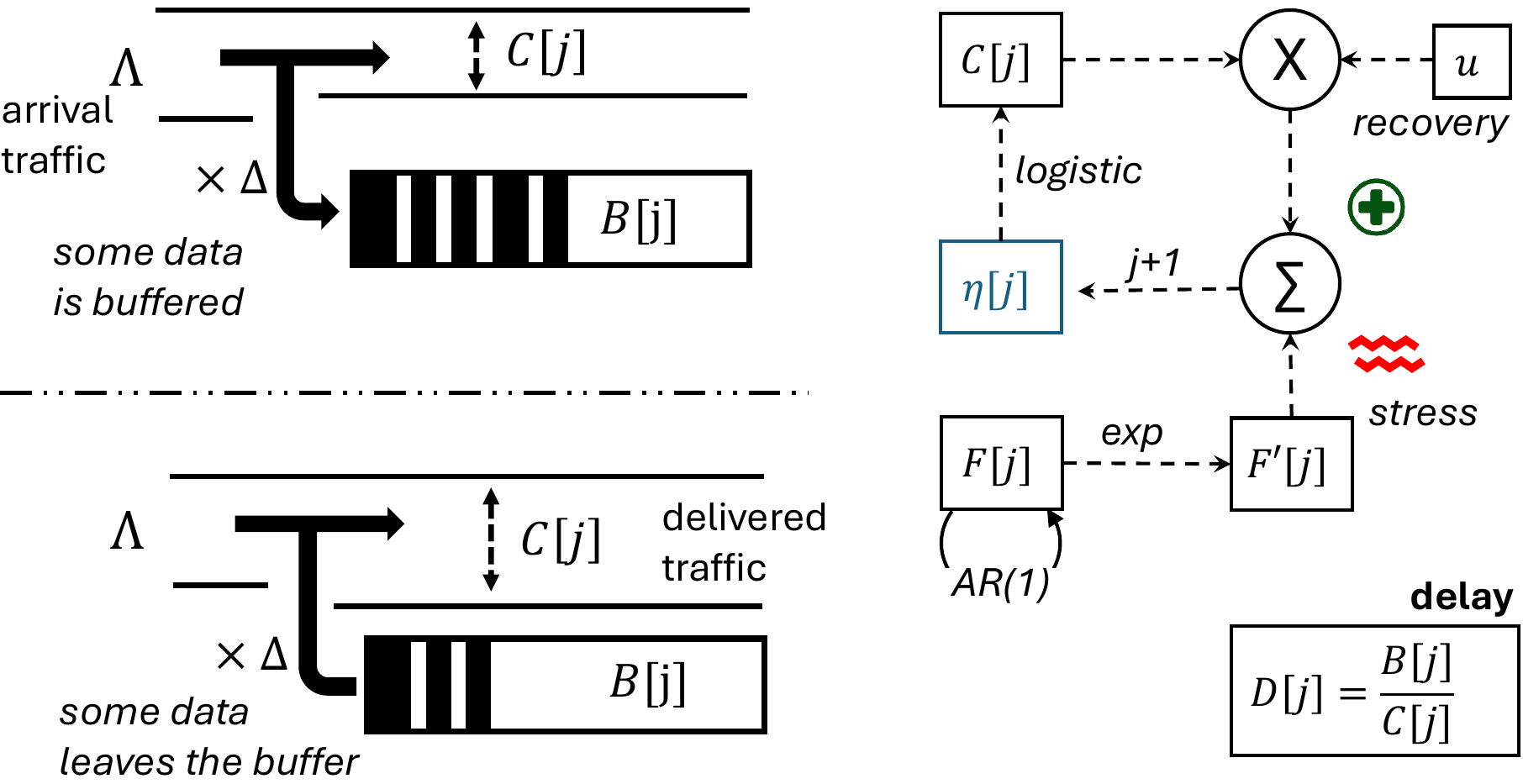}
    \caption{System dynamics of the use case.}
    \label{fig:bufferS}
    \vspace{-4mm}
\end{figure}
 
We assess resilience through the system's ability to recover from a (rare) significant delay increase within a prescribed time window. For this, let
\begin{align}
    j_c\triangleq \inf\{j: D[j]\ge \delta\}
\end{align}
denote the first time index at which the critical service delay threshold $\delta$ is exceeded. Once this event occurs, a recovery deadline of duration $t_{\rm tar}$ is imposed. 
Specifically, a resilience break occurs if, after the first critical affectation, the delay remains above the critical threshold continuously throughout the grace period. This is mathematically represented by
\begin{align}
    \xi \triangleq \{\exists j_c\le J-H: D[j]\ge \delta, \forall j\in[j_c,j_c+H]\}, \label{eq:xi2}
\end{align}
where $H\triangleq \lceil t_{\rm tar}/\Delta\rceil$ is the grace period expressed in discrete-time steps. 
Note that the event in \eqref{eq:xi2} captures a path-dependent notion of failure, as an initial disruption must occur, followed by insufficient recovery within the allowed time window.

%Together with the backlog dynamics in (22)–(23), the variables $B[j], \eta[j],$ and $Z[j]$ define a discrete-time stochastic system driven by disturbance arrivals, stress magnitudes, and recovery dynamics. In the following, we will be particularly interested in the one-step degradation increment $\Delta\eta[j]=\eta[j+1]-\eta[j]$, whose conditional distribution depends on the current health state and captures the stochastic impact of stress and recovery on system resilience.

%}

\subsection{Resilience Failure Probability Estimation}\label{sec:resR}
Let the system state at time step $j$ be defined as 
\begin{align}
    X[j]\triangleq (B[j],C[j], F[j], \varrho[j]),\label{eq:X}
\end{align}
where $B[j]$, $C[j]$, and $F[j]$ evolve according to \eqref{eq:B}$-$\eqref{eq:relax}, while 
$\varrho[j]\in\{0,1,\cdots,H\}$ is a persistence variable counting the number of consecutive time steps for which the delay has remained above the critical threshold $\delta$. The latter evolves according to the recursion
%an auxiliary state variable representing the remaining recovery slack. The latter is introduced to preserve the Markov property of the augmented process and is defined as 
%
\begin{align}
   % \varrho[j]&=\left\{\begin{array}{ll}
   %      0, & j\le j_c \\
   %      j-j_c,& j_c<j< j_c+H\\
   %      H, &  j_c+H\le j \le J
   % \end{array}\right.\nonumber\\
   % &=\min(\max(0,j-j_c),H),\ \forall j\le J.
   \varrho[j+1]\!=\!\left\{\begin{array}{ll}
        \min(\varrho[j]+1,H), & D[j]\ge \delta  \\
        0, & D[j]<\delta
   \end{array}\right.\!\!,\ \forall j<J,\label{eq:z}
\end{align}
with the initial condition $\varrho[0]=0$, and where $D[j]$ is given by \eqref{eq:delay}.
%Hence, it is zero while the system operates normally or returns to a normal state, and starts increasing upon service degradation, reaching the maximum value $H$ when the recovery deadline expires. 
This construction ensures that i) $\varrho[j]=d$ iff the delay has remained continuously above $\delta$ for the last $d$ time steps, and ii) $\{X[j]\}$ is a time-homogeneous Markov process.
With this in place, the resilience failure event $\xi$ in \eqref{eq:xi2} is equivalently the hitting of a failure set $\mathcal{F}=\{(B,C,\varrho): \varrho=H\}$ as
\begin{align}
    \xi\triangleq \{\exists j: \varrho[j]=H\}, 
\end{align}
%with first-hitting time index given by
%
%\begin{align}
%    j_\mathcal{F}\triangleq \inf\{j: \varrho[j]=H\}.
%\end{align}
and the goal is to estimate $\Pr(\xi)$.

A meaningful choice for the reaction coordinate is
\begin{align}\label{eq:react_cord}
    g(X[j])=\min\Big(\frac{D[j]}{\delta},1\Big) + \frac{\varrho[j]}{H}\ \in[0,2],
\end{align}
as it captures the end-to-end delay up to the critical deadline $\delta$, shifting then the focus to the recovery delay. This mapping captures the resilience fault progression in an ordered and delimited manner, and allows introducing a sequence of increasing levels for $g(\cdot)$ as
%
%\begin{align}
    $0=\ell_0<\cdots<\ell_K=2$.
%\end{align}
%
A simple choice for the intermediate levels is $\ell_1=1/10$, $\ell_2=1$, and $\ell_3=3/2$, which respectively capture initial progression towards critical service affectation, the start of a critical service affectation, and mid-progression towards non-recovery. This formulation with $K=4$ allows tracking the path toward failure in a structured manner.

\begin{table}[t]
\centering
\caption{Baseline parameters for simulation.}
\label{tab:baseline_params}
\begin{tabular}{lll}
\toprule
\textbf{parameter} & \textbf{description}& \textbf{value / range} \\
\midrule
$\Delta$ & discrete-time step & $50$ ms \\
$T$ & simulation horizon & $60$ s \\
%$C_0$ & initial/nominal service rate & 1 (normalized)\\
$B[0]$ & initial backlog condition & $0$\\
$\Lambda$ & offered traffic rate & 0.7 \\
%$D_0$ & baseline delay & $20$ ms \\
%$\lambda(u_p)$ & stress event rate & 0.1 $\text{s}^{-1}$\\
%$S_\text{max}$ & maximum capacity drop & 0.1 \\
$\eta[0]$ & initial/nominal recovery rate & $0.95$\\
%$\eta_0$ & capacity floor & 0.1\\
%$p_{01}$ & calm-to-noisy transition probability & 0.1 \\
%$p_{10}$ & noisy-to-calm transition probability & 0.1 \\
%$\lambda_0$ & low stress-event intensity & $1$ s$^{-1}$ \\
%$\lambda_1$ & high stress-event intensity & $5$ s$^{-1}$ \\
%$S_1$ & capacity drop per event & $\mathcal{U}(0,s_\text{max})$ \\
%$\psi,\vartheta$ & background disturbance parameters & $1, 2$\\
%$\underline{\lambda}$ & residual stress lower bound & $>0$ \\
%$u_p$ & prevention/containment control & ? \\
%$u$ & recovery-acceleration control & ? \\
$\nu(u)$ & recovery (relaxation) rate & $0.2$ \\
$\varphi(u)$ & non-linearity recovery parameter & 2\\
$\rho$ & temporal correlation of the latent stress & 0.75\\
$\mu_F$ & long-term mean of the latent stress  & -5 \\
$\sigma_F$ & standard deviation of the latent stress  & $0.55$ \\
%$\rho'$ & recovery rate policy multiplier & 0.5 \\
$\delta$ & service degradation threshold & 100 ms \\
%$D_M$ & massive impact threshold & 200 ms \\
$t_{\rm tar}$ & recovery target time & 5 s\\
$K$ & number of SMC levels & 4\\
$\ell_0,\cdots,\ell_4$ & SMC levels' thresholds & $0,0.1,1,1.5,2$\\
%$M$ &  SMC particles per level & 2000 \\
%$\varepsilon$ & numerical stability constant & $10^{-6}$ \\
$S_{\rm tar}$ & target number of successful crossings & 20 \\
$A_{\rm tar}$ & min. attempted continuations per level & 100 \\
$M_m,M_M$ & min./max particle pool size per level & 20, 200 \\
$p_{\rm min}$ & min. probability floor for stability & 0.05 \\
$\varsigma$ & safety factor scaling population sizes & 1.5 \\
$C_T$ & total simulation cost & $5\times 10^6$\\
\bottomrule
\end{tabular}
\end{table}

We illustrate some results in Fig.~\ref{fig:assessment} using the parameter values listed in Table~\ref{tab:baseline_params}, which are the defaults hereafter unless otherwise specified. We show the estimated probability of the resilience failure event as a function of the service degradation threshold for both MC and SMC simulation frameworks under the same computational cost, the latter using the procedure in Alg.~\ref{alg:budget_smc}. For this, we assume no active or adaptive mitigation so that the control parameter $u_p$ 
%and $u$ are 
is held constant. 
The results confirm the superiority of SMC over MC in the region of extreme/rare failure events, i.e., $\Pr(\xi)\le 10^{-3}$. Indeed, MC can only simulate $C_{\rm mc}/J=5\times 10^{6}/(60\ \text{s}/50\ \text{ms})\approx 4167$ trajectories, limiting its fault estimate to values greater than $1/4167=2.4\times 10^{-4}$. For SMC, a conservative resolution scale is $(1/A_{\rm tar})^K$, assuming each stage uses at least $A_{\rm tar}$ attempts, hence making it feasible to estimate fault probabilities at least greater than $10^{-8}$. 
%This explains why SMC can resolve probabilities below the MC resolution under the same step budget.
%an exact analysis is cumbersome, but it is clear from Fig.~\ref{fig:assessment} how it can produce estimates for several orders of lower fault probabilities.

\begin{figure}
    \centering
    \includegraphics[width=0.95\linewidth]{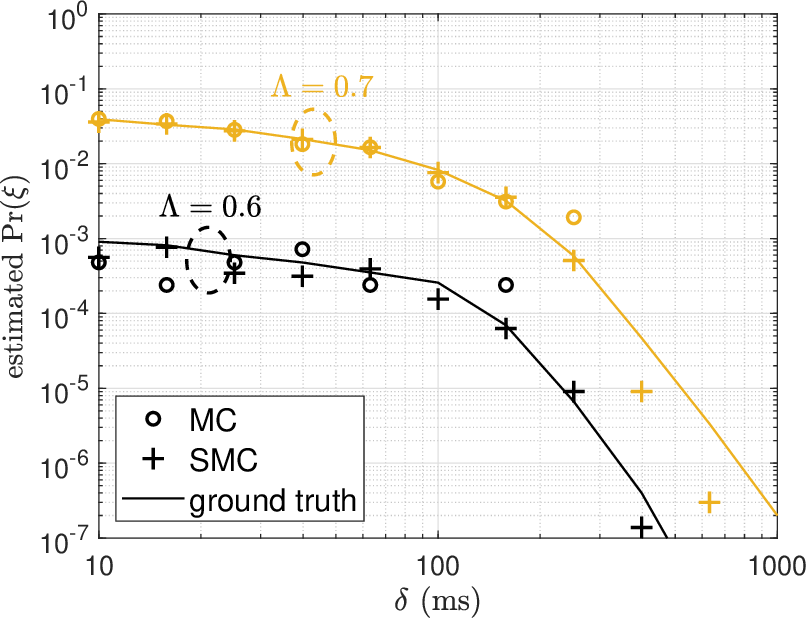}
    \caption{Estimated non-recovery probability $\Pr(\xi)$ versus the service degradation threshold $\delta$ for $\Lambda\in\{0.6,0.7\}$ ms and using MC and SMC simulation approaches. 
    %The ground truth is estimated using SMC with a greater $M$ and $C$
    %We set $S_\text{max}=\Lambda/5$ so that stress events emerging from higher traffic and capacity drops per event are positively correlated and drive resilience failure.
    }
    \label{fig:assessment}
\end{figure}

\subsection{Service Reconfiguration}
Now, we showcase SMC for fault management as discussed in Section~\ref{sec:reconf}. For simplicity, we consider service reconfiguration actions $u$ that are activated only after the system enters a critically degraded operating regime, corresponding to the first exceedance of the service delay threshold $\delta$. This corresponds to the system reaching level $\ell_2$.

We consider a finite set of candidate recovery policies $\mathcal{U}=\{u_0,\cdots,u_{|\mathcal{U}|-1}\}$, where each policy $u_i$ leads to a distinct recovery rate 
\begin{align}
    \nu(u_i)=\nu_0(1+i\rho'),\ \rho'\in(0,1],
\end{align}
with $\nu_0$ denoting the baseline recovery rate used prior to any reconfiguration and corresponding to $u_0$, and such that $\nu(u_{|\mathcal{U}|-1})\Delta\le 1$. 
%Then, to select among candidate recovery policies upon entering the critical regime, we
We adopt the myopic policy selection rule introduced in Section~\ref{sec:reconf}, such that \eqref{eq:LJ} becomes
\begin{align}
    \mathcal{J}(u_i|x_2^\star)=\ln \hat{p}_2(u_i) + c_2(u_i). \label{eq:LJ2}
\end{align}
Also, we model the policy cost proportional to the relative acceleration of recovery, as
\begin{align}
    c_2(u_i)=\kappa (\nu(u_i)-\nu_0)/\nu_0,\label{eq:cost}
\end{align}
where $\kappa\ge 0$ is a scaling parameter regulating the trade-offs between both terms in \eqref{eq:LJ2}.

%\subsubsection{SMC-assisted reconfiguration}
The conditional probability $\hat{p}_2(u_i)$ is estimated via simulation by exploiting the SMC checkpoints at level $\ell_2$. Specifically, upon first reaching $\ell_2$, the corresponding checkpoint state $x_2^\star$ is used as a common starting point, from which $N'$ independent trajectory continuations are generated under policy $u_i$, each with fresh realizations of the stochastic stress and recovery dynamics. Then, the policy that provides the smallest $\mathcal{J}(\cdot)$ in \eqref{eq:LJ2} is selected.
 %We assume that the selected policy 
Simulations for the estimation of the resilience fault probability are continued then from $\ell_2$ onward with new random processes and the selected policy in place until the end of the horizon.

\begin{figure}
    \centering
    \includegraphics[width=0.95\linewidth]{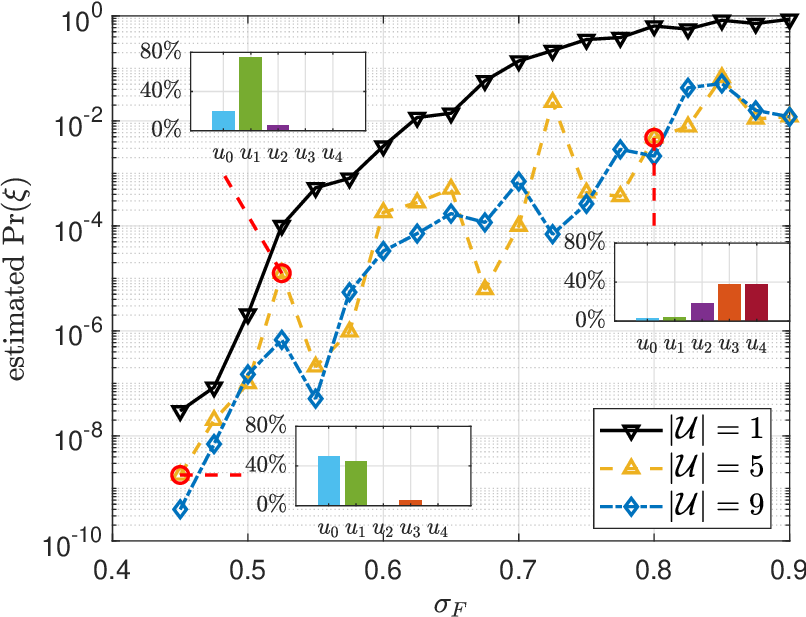}%\\
    \caption{Estimated non-recovery probability $\Pr(\xi)$ versus the standard deviation of the latent stress $\sigma_F$ with the proposed SMC-assisted reconfiguration framework and policy sets $\mathcal{U}$ of different dimensions. 
    %b) Relative frequency of utilization for the policies for the case of $|\mathcal{U}|=5$ (bottom). 
    We set $\rho'=1/2$, $\kappa=1/2$, and $N'=25$. For the case of $|\mathcal{U}|=5$, we also plot the relative selection frequency of each policy for $\sigma_F\in\{0.45, 0.575, 0.8\}$.}
    \label{fig:reconf}
\end{figure}

%\subsubsection{Numerical results}
Fig.~\ref{fig:reconf} illustrates the behavior of the estimated resilience fault probability and policy selection as a function of the latent stress variability across different dimensions of the policy sets under the proposed SMC-assisted reconfiguration framework. For this, we consider no policy reconfiguration, i.e., $|\mathcal{U}|=1$, as in the case of Section~\ref{sec:resR}, and also some reconfigurability using $|\mathcal{U}|=5$ and $|\mathcal{U}|=9$. 
We can corroborate that more comprehensive policy sets, i.e., larger $|\mathcal{U}|$, allow for more flexible recovery-cost trade-offs, often leading to improved recovery while still accounting for the policy costs \eqref{eq:cost}. Notably, under low stress, i.e., small $\sigma_F$, lower-rank policies are preferred due to their lower cost, while higher-rank, hence more disruptive, policies gain more traction as stress increases. 
All this reveals the need for carefully assessing policy selection in resilience frameworks, as consistently aggressive handling of faults can incur costs that may affect future system states.
%In highly stressed regimes, e.g., $\sigma_F\ge 0.8$, the most radical policy is adopted more than $38\%$ of the time, but surprisingly, the second most adopted is the least radical policy, with a frequency above $25\%$. All this reveals the need for carefully assessing policy selection in resilience frameworks, as consistently aggressive handling of faults can incur costs that may affect future system states.
\subsection{SMC meets genAI}

We replace the system simulator with a learned stochastic surrogate. Since different criticality levels correspond to different stages of the same delay-evolution process, 
we adopt a single context-conditioned modeling as discussed in Section~\ref{sec:AI}. 
%allowing one shared generator to adapt its future continuations to the degradation level without introducing model switches across regimes. 
The conditional trajectory law is learned with a denoising diffusion probabilistic model (DDPM).
%that generates an entire future delay block conditioned on the recent delay history and a criticality context.

\subsubsection{DDPM surrogate}

The surrogate learns a stochastic transition model for delay windows. 
%Since delay values can vary over several scales, and 
To reduce the dynamic range and improve the numerical stability, we first apply a log-delay transformation  given by
\begin{equation}
y[j]
=
\log\left(1+{D[j]}/{s_D}\right),
\end{equation}
where $s_D>0$ is a fixed scaling constant. The log-delay samples are then normalized using the mean and standard deviation obtained from the training data. Each training example is constructed as $\mathbf{y}^{\mathrm{h}}_j, \mathbf{y}^{\mathrm{f}}_j $ by using \eqref{eq:futhist_concat} and
we define ${z}_j = {g(X[j])}/{2}$ as the criticality context.
%, where $g(X[j])$ is computed by \eqref{eq:react_cord}.

As per Section~\ref{sec:AI}, the learned stochastic trajectory law is specified by
$
p_\theta
\left(
\mathbf{y}^{\mathrm{f}}_j
\mid
\mathbf{y}^{\mathrm{h}}_j,{z}_j
\right)$, which we construct using: i) a gated recurrent unit (GRU) that encodes the delay-history window $\mathbf{y}^{\mathrm{h}}_j$ into a temporal feature vector, ii) a multi-layer perceptron (MLP) that embeds the criticality context ${z}_j$, and iii) another MLP that embeds the last observed delay value $y[j]$, providing an anchor between the observed history and the generated future block. These three representations are concatenated and passed through a fusion MLP to obtain the final conditioning vector $\mathbf{h}_j$.

% For this, the conditioning representation is obtained as
% \begin{align}
% \mathbf{h}^{y}_t
% &=
% \mathrm{GRU}_{\phi_y}
% \left(
% \mathbf{y}^{\mathrm{h}}_t
% \right), \mathbf{h}^{c}_t
% =
% \mathrm{MLP}_{\phi_c}
% \left(
% {\mathbf{c}}_t
% \right),\\
% \mathbf{h}^{a}_t
% &=
% \mathrm{MLP}_{\phi_a}
% \left(
% y[t]
% \right),
% \mathbf{h}_j
% =
% \mathrm{MLP}_{\phi_f}
% \left(
% [\mathbf{h}^{y}_t,\mathbf{h}^{c}_t,\mathbf{h}^{a}_t]
% \right).
% \end{align}
% Here, $\mathbf{h}^{y}_t$ encodes the recent delay history, $\mathbf{h}^{c}_t$ embeds the criticality context, and $\mathbf{h}^{a}_t$ embeds the last observed delay value. The last component acts as an anchor that helps align the generated future block with the observed history.

The generator is a conditional DDPM over the whole future block $\mathbf{y}^{\mathrm{f}}_j$ \cite{ho2020denoising}. Here, the forward diffusion process is given by
\begin{equation}
\mathbf{y}_{\tau}
=
\sqrt{\bar{\alpha}_{\tau}}\,
\mathbf{y}^{\mathrm{f}}_j
+
\sqrt{1-\bar{\alpha}_{\tau}}\,
\boldsymbol{\epsilon},
\qquad
\boldsymbol{\epsilon}\sim\mathcal{N}(\mathbf{0},\mathbf{I}),
\end{equation}
where $\tau$ is sampled uniformly from the diffusion time index set, and $\bar{\alpha}_{\tau}$ follows a cosine noise schedule. The denoising network is a one-dimensional U-Net operating along the future-time axis and predicts the injected noise as
\begin{equation}\label{eq:noisepred}
\widehat{\boldsymbol{\epsilon}}
=
\epsilon_\theta
\left(
\mathbf{y}_{\tau},\tau,\mathbf{h}_j
\right).
\end{equation}
The diffusion timestep is represented through a sinusoidal embedding followed by an MLP. This timestep embedding and the conditioning vector $\mathbf{h}_j$ are injected into the residual convolutional blocks by feature-wise linear modulation.

The standard DDPM noise-prediction loss is given by
\begin{equation}
\mathcal{L}_{\mathrm{DDPM}}
=
\mathbb{E}_{j,\tau,\boldsymbol{\epsilon}}
\left[
\frac{1}{L_f}
\sum_{r=1}^{L_f}
\big(
\epsilon_{r}
-
\epsilon_\theta
\left(
\mathbf{y}_{\tau},
\tau,
\mathbf{h}_j
\right)_{r}
\big)^2
\right].
\end{equation}
where $\epsilon_r$ denotes the $r$-th component of the injected noise vector, and the denoising network is trained to recover this noise from the corrupted future block at diffusion timestep $\tau$. In the implementation, the loss term is reweighted to improve sensitivity to resilience-relevant samples. Specifically, windows associated with higher criticality levels receive larger sample weights, and future samples whose clean delay exceeds the critical threshold receive larger per-time-step weights. This does not change the DDPM training principle and only biases the finite-sample objective toward critical parts of the trajectory distribution, which are most relevant for estimating resilience failure probabilities.

Although the trained diffusion model defines a full DDPM reverse chain, using all diffusion steps for every generated block would be computationally expensive during MC/SMC sampling. Therefore, we use a denoising diffusion implicit model (DDIM)-style sampler \cite{song2020denoising} to accelerate inference while reusing the same trained denoising network. Specifically, we use a subsampled DDIM reverse schedule with $T'_d$ steps, where $\tau_S$ is the largest selected timestep from the full DDPM schedule of length $T_d$. The reverse chain is initialized around the last normalized history value as
\begin{equation}
\mathbf{y}_{\tau_S}
=
y[j]\mathbf{1}
+
\boldsymbol{\upsilon},
\qquad
\boldsymbol{\upsilon}\sim\mathcal{N}(\mathbf{0},\mathbf{I}).
\end{equation}
For two consecutive timesteps $\tau$ and $\tau'$ in the subsampled reverse schedule, with $\tau'<\tau$, the denoiser first predicts the injected noise $\widehat{\boldsymbol{\epsilon}}_\tau$ by \eqref{eq:noisepred}, and the next reverse state is obtained using the standard DDIM update \cite{song2020denoising}. We use the fully stochastic variant of DDIM such that repeated sampling from the same history and criticality context can produce diverse continuations. Finally, the sampled normalized log-delay block is mapped back to raw delay.

% For two consecutive timesteps $\tau$ and $\tau'$ in the subsampled reverse schedule, with $\tau'<\tau$, the denoiser first predicts the noise $\widehat{\boldsymbol{\epsilon}}_\tau$ by \eqref{eq:noisepred}. The DDIM update is then given by
% \begin{multline}
% \mathbf{y}_{\tau'}
% =
% \sqrt{\bar{\alpha}_{\tau'}}
% \left(
% \frac{
% \mathbf{y}_{\tau}
% -
% \sqrt{1-\bar{\alpha}_{\tau}}\,
% \widehat{\boldsymbol{\epsilon}}_\tau
% }{
% \sqrt{\bar{\alpha}_{\tau}}
% }
% \right) \\
% +
% \sqrt{1-\bar{\alpha}_{\tau'}-\sigma_{\tau,\tau'}^{2}}\,
% \widehat{\boldsymbol{\epsilon}}_\tau
% +
% \sigma_{\tau,\tau'}\mathbf{n},
% \end{multline}
% where $\mathbf{n}\sim\mathcal{N}(\mathbf{0},\mathbf{I})$ and
% \begin{equation}
% \sigma_{\tau,\tau'}
% =
% \eta
% \sqrt{
% \frac{1-\bar{\alpha}_{\tau'}}
%      {1-\bar{\alpha}_{\tau}}
% \left(
% 1-
% \frac{\bar{\alpha}_{\tau}}
%      {\bar{\alpha}_{\tau'}}
% \right)
% }.
% \end{equation}
% Thus, $\eta$ directly scales the noise injected at each reverse step. Setting $\eta=0$ removes the stochastic term and yields deterministic DDIM sampling, while $\eta>0$ produces stochastic continuations. In this work, we set $\eta=1$ to generate diverse stochastic continuations. Finally, the sampled normalized log-delay block is mapped back to raw delay.

% We consider two dataset construction modes. The first, \textit{uniform}, samples valid post-transient windows uniformly using the predefined stride. The second mode, \textit{balanced},

\subsubsection{Dataset construction}

We follow Section~\ref{sec:criticality_data} and design the dataset to improve the representation of critical and transitional regimes by combining two sets of windows. The first set is obtained from a balanced criticality-aware selection. Since most valid checkpoints belong to the lowest criticality level, $k=0$, these windows are sampled using a stride $N_s$. In contrast, all valid checkpoints with $k>0$ are retained as candidates to improve the representation of degraded and near-critical behavior. The number of retained windows per level is capped using a frequency-aware rule. Level $0$ receives cap $\hat{N}$, while level $k>0$ receives
\begin{equation}
\hat{N}_{k}
=
\max
\left\{
N_{\min},
\left\lfloor
\hat{N}
\frac{f_k}{f_0}
\right\rfloor
\right\},
\end{equation}
where $f_k$ is the empirical frequency of level $k$ over checkpoints, $f_0$ is the level-$0$ frequency, and $N_{\min}$ is a minimum cap. The second set consists of transition windows, defined as valid checkpoints where the discrete level index changes between two consecutive time instants. These transition windows include both upward and downward level changes and are capped at $N_t$ independently per level. The final dataset is the concatenation of the balanced and the transition windows.

% To isolate the effect of window selection from the effect of dataset size, \textit{uniform} is size-matched to \textit{balanced}; if more uniform candidates are available, a random subset of the required size is retained. 

\subsubsection{Numerical results}

% \begin{table}[t]
% \centering
% \caption{Parameters for the diffusion-DT.}
% \label{tab:diffusion_dt_params}
% \begin{tabular}{llll}
% \hline
% \textbf{Parameter} & \textbf{Value} & \textbf{Parameter} & \textbf{Value} \\
% \hline
% Source trajectories & $20000$ & $L_h$ & $64$ \\
% $L_f$ & $16$ & $s_D$ & $0.05$ \\
% Low-criticality stride & $10$ & $\hat{N}$ & $15000$ \\
% $N_{\min}$ & $1000$ & Transition windows & $3000$ \\
% Hidden dimension & $128$ & GRU layers & $1$ \\
% $T_d$ & $1000$ & DDIM steps & $50$ \\
% U-Net base channels & $64$ & Epochs & $100$ \\
% Batch size & $128$ & Learning rate & $3\times 10^{-4}$ \\
% \hline
% \end{tabular}
% \end{table}

\begin{table}[t]
\centering
\caption{Parameters for the diffusion-DT.}
\label{tab:diffusion_dt_params}
\begin{tabular}{lll}
\toprule
\textbf{parameter} & \textbf{description} & \textbf{value} \\
\midrule
$N_{\rm src}$ & source trajectories for training data & $20000$ \\
$L_h$ & history-window length & $64$ \\
$L_f$ & future-block length & $16$ \\
$s_D$ & delay log-transform scale & $0.05$ \\
$N_s$ & low-criticality sampling stride & $10$ \\
$\hat{N}$ & balanced level-$0$ window cap & $15000$ \\
$N_{\min}$ & minimum window cap per level & $1000$ \\
$N_t$ & transition-window cap per level & $3000$ \\
$T_d$ & diffusion noise-schedule length & $1000$ \\
$T'_d$ & DDIM reverse-sampling steps & $50$ \\
\bottomrule
\end{tabular}
\end{table}

For each load value, the number of MC samples is chosen such that the total number of generated time steps matches the number of generated time steps used by the SMC procedure with $1000$ particles, ensuring equal inference cost across the compared sampling approaches. The first $20$~s of each trajectory are discarded before training and evaluation to remove initialization-dependent transient states. For surrogate evaluation, the physical simulator is used only to initialize the first $L_h$ delay samples. Thereafter, the diffusion DT is rolled out independently by generating future delay blocks, updating the persistence counter and reaction coordinate from its own samples, and continuing until the horizon ends or the stopping condition is met. The diffusion surrogate uses a hidden dimension of $128$, one GRU layer, $64$ U-Net base channels, $100$ training epochs, a batch size of $128$, and a learning rate of $3\times 10^{-4}$. The remaining learning and dataset parameters are specified in Table~\ref{tab:diffusion_dt_params}.

\begin{figure}
    \centering
    \includegraphics[width=0.95\linewidth]{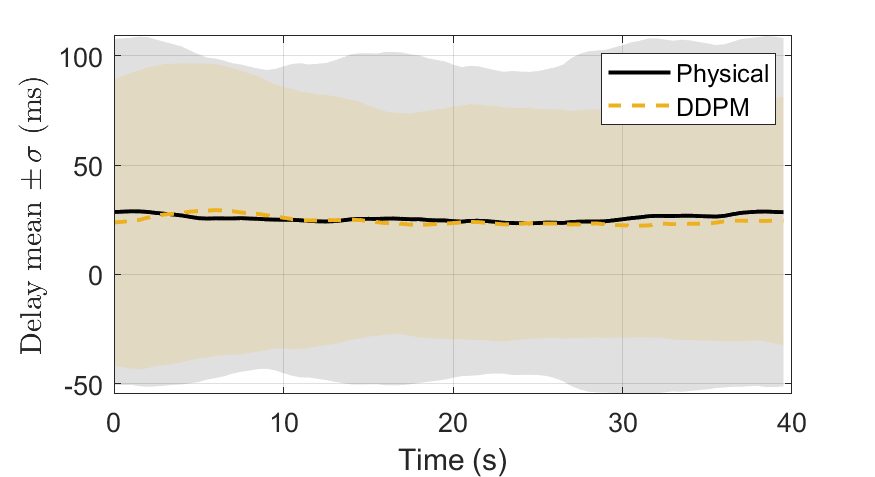}\\
    \includegraphics[width=0.95\linewidth]{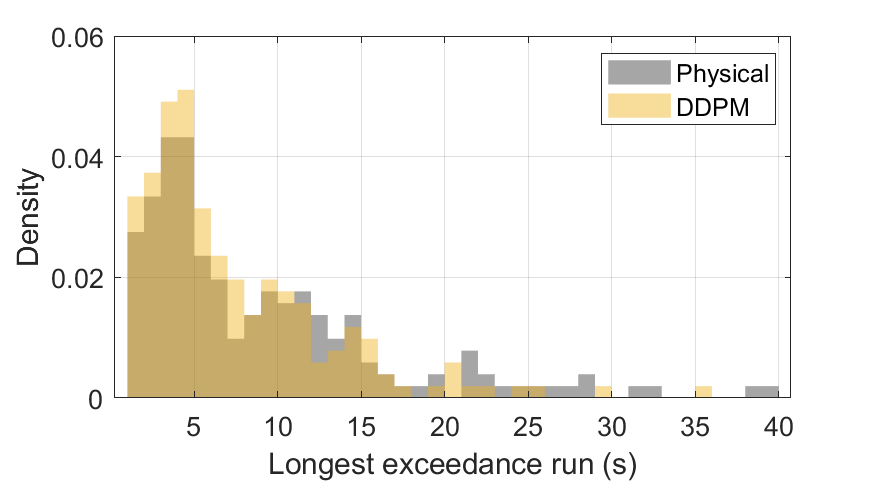}
    \caption{Comparison between physical-model and DDPM-generated delay trajectories for $\Lambda = 0.76$. (a) delay mean and standard deviation over the post-transient evaluation horizon (top) and (b) distribution of the longest threshold-exceedance run (bottom).}
    \label{fig:surrogate_stats}
    \vspace{-4mm}
\end{figure}

Fig.~\ref{fig:surrogate_stats} compares trajectory-level delay statistics obtained from the physical simulator and the DDPM surrogate. The upper panel shows that the surrogate reproduces the temporal evolution of the delay mean and standard deviation over the post-transient horizon. This indicates that, after being initialized with the first $L_h$ physical delay samples, the surrogate can generate delay trajectories without further access to the physical model while preserving the main marginal statistics. The lower panel compares the empirical distribution of the longest continuous threshold-exceedance duration in each trajectory. Note that in this figure, the large peaks of short exceedance runs are removed to better reveal the structure of the rest. This statistic is directly relevant to resilience assessment because failure is also driven by how long the exceedance persists. The DDPM surrogate captures the main shape of this distribution, suggesting that it can represent persistence patterns associated with critical delay events.

\begin{figure}
    \centering
    \includegraphics[width=0.95\linewidth]{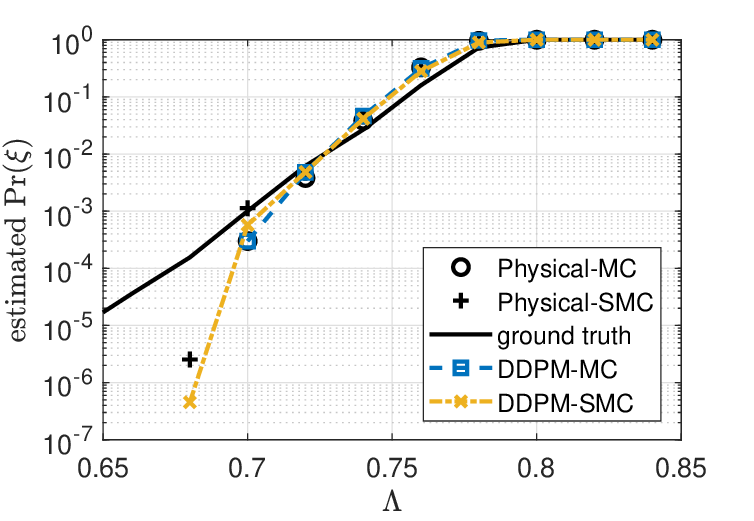} 
    \caption{Estimated non-recovery probability $\Pr(\xi)$ versus the offered traffic rate $\Lambda$ using MC and SMC approaches based on physical and surrogate models. }
    \label{fig:surrogate_results}
    \vspace{-4mm}
\end{figure}

Fig.~\ref{fig:surrogate_results} compares the estimated non-recovery probability obtained from the physical simulator and the DDPM surrogate using MC and SMC sampling. The DDPM estimates closely follow the physical-model baselines across the load range. The benefit of SMC is most visible in the low-probability region, where DDPM-MC is limited by the finite number of generated trajectories, while DDPM-SMC can estimate rarer events by branching from intermediate critical states. This proves that the learned diffusion surrogate can support SMC-style resilience assessment by generating delay continuations.

\section{Conclusions and Outlook}\label{sec:conclude}
%
%This work presented SMC for resilience assessment and control in networked systems. By formulating resilience failure events as path-dependent conditions over staged system dynamics, we showed how SMC enables efficient estimation of rare non-recovery probabilities through a multilevel splitting approach. The use of fixed, semantically meaningful levels allows for an interpretable decomposition of resilience into successive phases, while the proposed budget-adaptive population control mechanism ensures that computational effort is dynamically concentrated on the most critical transitions under a fixed total computation cost. The framework was also extended beyond estimation to support mitigation and control. This leverages SMC's checkpointing structure for policy evaluation and comparison at intermediate resilience levels, and for state-contingent policy selection via simulation-based lookahead.
This paper developed an SMC framework for resilience assessment and control in wireless networks. Resilience failures were modeled as path-dependent rare events associated with staged degradation and non-recovery, while SMC leveraged budget-aware population-control fixed-level splitting to provide both efficient sampling and interpretable stage-wise vulnerability measures. Moreover, SMC checkpoints were reused for mitigation policy evaluation and for local simulation-based policy selection.
The framework was also extended toward data-driven DTs by using generative sequence models as restartable stochastic surrogates and exploiting SMC mechanisms. We considered a delay-critical wireless-network use case, in which rare non-recovery probability estimation, SMC-assisted service reconfiguration, and diffusion-based surrogate simulation were demonstrated. The results showed that SMC can not only significantly outperform naive MC in rare-event regimes, but that this remains effective when the physical simulator is replaced by a learned stochastic surrogate.

%Future work will focus on assessing the potential of the proposed framework for more complex system dynamics and distributed control mechanisms. Also, beyond policy evaluation, we will consider the proposed framework as a data-generation mechanism for learning-based resilience control by exposing near-critical and post-fault states that are effectively unobservable under nominal sampling. Another relevant direction is the integration of data-driven models to approximate or accelerate the underlying simulator, enabling scalable resilience analysis in large-scale ?? systems.
%Investigating the integration of SMC-driven data generation with RL or other adaptive control techniques constitutes another promising direction for future work.

\bibliographystyle{ieeetr}
\bibliography{ref}

% that's all folks
\end{document}